\documentclass[12pt,a4paper]{article}

\usepackage{latexsym}
\usepackage{amssymb}
\usepackage{amsmath}
\usepackage{graphicx}
\usepackage{bbold}
\usepackage{mathtools}
\usepackage{hyperref}
\usepackage{todonotes}
\usepackage{xcolor}
\usepackage[small]{titlesec}

\usepackage{amsthm}
\usepackage{thmtools}
\usepackage{enumerate}
\usepackage{enumitem}
\usepackage{tikz}
\usepackage{graphicx}

\newcommand{\ket}[1]{\left| #1 \right\rangle}

\newcommand{\braket}[2]{\langle #1 | #2 \rangle}

\newcommand{\comment}[1]{}

\newcommand{\BraKet}[3]{\langle #1 |\, #2  \,| #3 \rangle}
\newcommand{\E}[1]{\langle #1 \rangle}
\newcommand{\ES}[2]{\langle #1 ; #2 \rangle}
\newcommand{\llangle}{\langle\!\langle}
\newcommand{\rrangle}{\rangle\!\rangle}
\newcommand{\EE}[1]{\llangle #1 \rrangle}

\newtheorem{thm}{Theorem}

\newtheorem{proposition}[thm]{Proposition}
\newtheorem{lemma}[thm]{Lemma}

\newtheorem{remark}[thm]{Remark}

\declaretheoremstyle[
  headfont=\color{red}\normalfont\bfseries,
  bodyfont=\color{red}\normalfont\itshape,
]{colored}

\theoremstyle{remark}

\theoremstyle{colored}

\theoremstyle{colored}

\theoremstyle{colored}

\numberwithin{equation}{section}

\DeclareMathOperator{\dd}{d}
\DeclareMathOperator{\ii}{i}
\DeclareMathOperator{\ee}{e}
\DeclareMathOperator{\im}{Im}
\DeclareMathOperator{\re}{Re}

\title{Indirect Measurements of a Harmonic Oscillator}
\author{Martin Fraas,
\\
\small{Mathematics, Virginia Tech, VA 24061, Blacksburg, U.S.A.}
\\Gian Michele Graf, Lisa H\"{a}nggli
\\
\small{Theoretische Physik, ETH Zurich, 8093 Zurich, Switzerland} }

\begin{document}
\maketitle

\begin{abstract}
The measurement of a quantum system becomes itself a quantum-mechanical process once the apparatus is internalized. That shift of perspective may result in different physical predictions for a variety of reasons. We present a model describing both system and apparatus and consisting of a harmonic oscillator coupled to a field. The equation of motion is a quantum stochastic differential equation. By solving it we establish the conditions ensuring that the two perspectives are compatible, in that the apparatus indeed measures the observable it is ideally supposed to.
\end{abstract}

\begin{section}{Introduction} \label{section_intro}
How long should a measurement last? Ideally, a measurement is instantaneous, since that is implicit in Born's rule. Such an idealization may not be appropriate in principle, because it takes time for the pointer of the apparatus to correlate with the intended observable of the system. It may nonetheless be appropriate effectively if the observable is a constant of motion, making the instant of the measurement irrelevant. Actually, and quite oppositely, the longer the measurement then takes the better the correlation gets established.

Quite often however observables get measured even if they are not constants of motion. It then appears consequential to think of the measurement as reflecting the observable averaged over the time taken by the measurement itself. Still, one should ask to which extent such a measurement is an estimator for the instantaneous measurement which, though impossible, remains of importance since it underlies Born's rule. In this paper we will discuss the available time window, which is limited from below by the need to correlate the system to the apparatus and from above by the back-reaction of the latter on the former.

These issues are discussed in a \emph{model}, to be described below, that is rich enough so that they are not trivial, yet simple enough that its dynamics can be solved for exactly. The model features an observable and a Hamiltonian which do not commute, and they do so to a degree tunable by a parameter $\alpha$, with $\alpha=0$ corresponding to a vanishing commutator and to a (so-called) non-demolition experiment. It moreover features an apparatus which, as appropriate for such devices, consists of a (macroscopically) large number of degrees of freedom. Collectively they realize a pointer which is supposed to meter the observable. In order to do its job, the (microscopic) degrees of freedom are coupled one by one to the system proper with coupling parameter $\gamma>0$, in guise of a repeated measurement. As a result, the pointer observable commutes at different times, in line with the classical nature of a record. 

Even before specifying the model in further detail, the following can be noted. Let $H$ be the \emph{Hamiltonian} of some physical system and let $O$ be some \emph{observable} one intends to measure indirectly. By this we mean that effectively the time averaged observable 
\begin{equation}\label{eq_averaged_obs}
\overline{O}_T=\frac{1}{T}\int_0^T \ee^{\ii Ht}O\ee^{-\ii Ht}\dd t\,,
\end{equation}
with $T$ large shall be measured projectively. If $O$ is the Hamiltonian itself the time average is redundant, $\overline{H}_T=H$, consistently with the fact that the projective measurement of $O=H$ is of the non-demolition type.

The situation discussed in this article is that of $O=H_0$, where $H_0$ is a reference Hamiltonian (or $O$ an observable affiliated with $H_0$) w.r.t.\ which the actual Hamiltonian $H$ is a perturbation. Such a situation can e.g.\ occur because full experimental control of the Hamiltonian is lacking. 

The effect is twofold: First $\overline{(H_0)}_T$ differs from the intended observable $H_0$, though the two operators still agree in expectation for eigenstates of $H$; second, the apparatus steadily demolishes the eigenstates of $H$, thereby heating up the system. The first error is independent of $T$ on eigenstates and oscillatory on (discrete) supersitions thereof and in any event as small as the control of the perturbation of $H-H_0$ allows. By the second error however the measurement no longer reflects the properties of the initial state, at least eventually.

As we will show in the context of the model, the second error is comparatively small. More precisely the measurement time $T$ can still be taken large enough so as to be able to tell eigenstates of $H_0$ apart, yet also small enough so as they do not become blurred by the heating.

The model is as follows. The Hamiltonian of the system proper are harmonic oscillators 
\begin{align*}
H_0&=\omega a^*a\equiv \omega N\,,\\
H&=H_0-\omega\left(\overline{\alpha} a +\alpha a^*\right)\\
&=\omega\left(\left(a^*-\overline{\alpha}\right)\left(a-\alpha\right)-|\alpha|^2\right)\,,
\end{align*}
($[a,a^*]=1$, $\alpha\in \mathbb{C}$), and the observable to be discussed shall be the excitation number $N$. The measurement apparatus will be realized later as a quantum field.

The motion generated by $H$ can be visualized classically as clockwise circular orbits in the complex $a$-plane centered at $\alpha$,
\begin{equation*}
a=\alpha+r\ee^{-\ii \omega t}\,,
\end{equation*}
($\alpha$-circles, see Fig.\ \ref{figure_alpha_circles}). As a result the average value of $N=a^*a$ is
\begin{equation}\label{eq_N_T_infty_c}
\overline{N}_T \to |\alpha|^2+r^2\,,\quad (T\to \infty)\,,
\end{equation}
because the mixed terms are oscillatory. Eq.\ (\ref{eq_N_T_infty_c}) has the following semi-classical interpretation: The Fock state $\ket{n}$, ($n\in \mathbb{N}$) is associated with a $0$-circle $|a|=r$ of square radius $r^2=n$. Its points lie on $\alpha$-circles of different radii, $r^2=|a-\alpha|^2$, and average value $\langle r^2\rangle =|a|^2+|\alpha|^2=n+|\alpha|^2$, leading by (\ref{eq_N_T_infty_c}) to 
\begin{equation}\label{eq_N_T_infty_n_c}
\BraKet{n}{\overline{N}_T}{n}\to \BraKet{n}{\overline{N}_{\infty}}{n}=n+2|\alpha|^2\,,\quad (T\to \infty)\,.
\end{equation}
The time $T$ needed to approach the limit is $T\gg \omega^{-1}$.
\begin{figure}\label{figure_alpha_circles}
\centering
\includegraphics[scale=0.7]{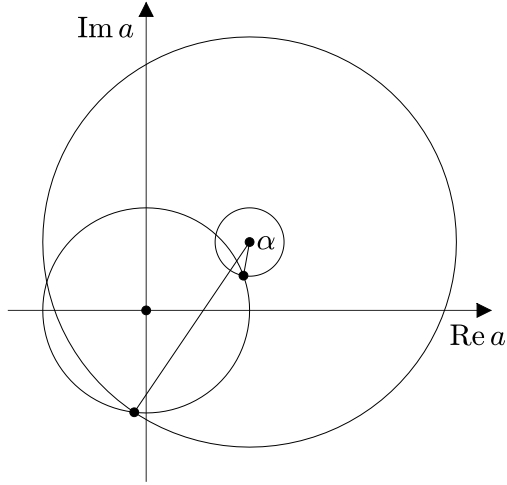}
\caption{$\alpha$-circles and $0$-circles.}
\end{figure}
Not surprisingly the quantum mechanical calculation confirms the above. In fact, and as shown in the Appendix,
\begin{equation}\label{eq_N_T_infty_qm}
\overline{N}_T\to \overline{N}_{\infty}=N-\left(\overline{\alpha}a+\alpha a^*\right)+2|\alpha|^2\,,
\end{equation}
which besides (\ref{eq_N_T_infty_n_c}) also yields
\begin{equation}\label{eq_N_T_infty_n_qm}
\langle \left(\Delta \overline{N}_{\infty}\right)^2\rangle_n=(2n+1)|\alpha|^2\,.
\end{equation}
We conclude that the standard deviation $\langle \left(\Delta \overline{N}_{\infty}\right)^2\rangle_n^{1/2}=\sqrt{2n+1}|\alpha|$ remains small compared to the spacing $\Delta \langle \overline{N}_{\infty}\rangle_n=1$ between the expectations of $\overline{N}_{\infty}$ in two consecutive states $\ket{n}$ and $\ket{n+1}$; at least for finitely many states $\ket{n}$ and provided $\alpha$ is small.

The issue to be investigated is the extent up to which that property persists when the measurement is in itself described as a dynamical process, thus including the feedback it exerts on the system proper. We do so by modelling the degrees of freedom of the apparatus by quantum noise. Full details will be given in the next section. For now it may suffice that the feedback results from the perpetual observation of the system during some time interval $[0,T]$, as it e.g.\ emerges from repeated measurements in the (non-trivial) limit where they become ever weaker yet ever more frequent. For an informal discussion it is best to postpone that limit. At the beginning of every time interval $\Delta t$ new noise degrees of freedom are introduced in a pristine state and are then coupled to the harmonic oscillator just for its duration. What we shall need is:
\begin{enumerate}[label=\roman*)]
\item The degrees of freedom of the noise are field quadratures $P_t$, $Q_t$, ($t\in[0,T]$) with commutation relations
\begin{equation*}
\ii [\Delta P_t,\Delta Q_t]=2\Delta t\,,
\end{equation*}
where $\Delta P_t$, $\Delta Q_t$ are the noises associated to $[t,t+\Delta t]$.
\item The Hamiltonian of the uncoupled apparatus is trivial. In the state $\ket{\Omega}$ of the apparatus $t\mapsto P_t$ has the same distribution as a Brownian motion on the real line with $\langle \left(\Delta P_t\right)^2\rangle=\Delta t$.
\item The coupling to the harmonic oscillator is  
\begin{equation}\label{eq_H_I}
H_I\Delta t=\gamma N\Delta P_t
\end{equation}
with $\gamma>0$. Upon passing to the Heisenberg picture ($O\rightsquigarrow \hat{O}_t=\ee^{\ii Ht}O\ee^{-\ii Ht}$, wherein $O$ may carry $t$ as a label in the Schr\"odinger picture) the change brought about on $\Delta \hat{Q}_t$ is 
\begin{equation*}
\delta \Delta \hat{Q}_{t} = \ii [\hat{H}_I\Delta t,\Delta \hat{Q}_{t}]\cong \ii [\hat{H}_I\Delta t,\Delta Q_{t}]= 2\gamma \hat{N}_t \Delta t\,
\end{equation*}
where $\cong$ refers to the leading order in $\Delta t$. The changes $\delta \Delta \hat{Q}_{t}$ are additive for different time intervals by (ii) and because the noise $\Delta Q_{t}$ of each is just transiently coupled to the oscillator. In particular, in the limit $\Delta t\to 0$, we have
\begin{equation}\label{eq_pointer}
\frac{\hat{Q}_{T}-Q_T}{T}=2\gamma \overline{N}_T
\end{equation}
whence $\mathcal{N}_T=\hat{Q}_{T}/2\gamma T$ eventually serves as a pointer for the observable $\overline{N}_T$, as intended. However by $[H_I,H]\not=0$ the apparatus potentially demolishes the eigenstates of $H$, as announced.
\end{enumerate}
That effect can again be discussed semi-classically, and we do so for simplicity in the weak coupling regime 
\begin{equation}\label{eq_wcr}
\gamma^2\ll \omega\,.
\end{equation}
We recall that the orbits of $H_0=\omega N$ are $0$-circles which run with frequency $\omega$. During a time $\Delta t$ the interaction (\ref{eq_H_I}) thus induces a turn by an angle $\Delta \psi=\gamma \Delta P_t$ along those circles. The resulting motion is diffusive with 
\begin{equation*}
\langle\left(\Delta\psi\right)^2\rangle=\gamma^2\Delta t
\end{equation*}
and is in competition with the drifting dynamics of $H$ which takes place along $\alpha$-circles (see Fig.\ \ref{figure_diffusive}). However during a period $2\pi/\omega$ we have 
\begin{figure}
\centering
\includegraphics[scale=0.7]{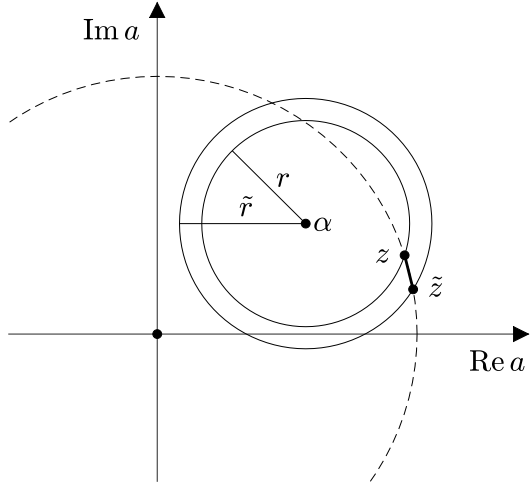}
\caption{Jump between $\alpha$-circles along a $0$-circle.}\label{figure_diffusive}
\end{figure}
%
\begin{equation*}
\langle\left(\Delta\psi\right)^2\rangle=2\pi\gamma^2\omega^{-1}\ll 1
\end{equation*}
by (\ref{eq_wcr}). Therefore the overall semi-classical motion follows some $\alpha$-circle interrupted by some rare jump along a $0$-circle to the next $\alpha$-circle. Let them be 
\begin{gather*}
z(\varphi)=\alpha+r\ee^{\ii \varphi}\,,\quad \tilde{z}(\tilde{\varphi})=\alpha+\tilde{r}\ee^{\ii \tilde{\varphi}}
\end{gather*}
with radii $r$, $\tilde{r}$, and be connected by the jump $z\mapsto \tilde{z}=z\ee^{-\ii \Delta \psi}$ which occurs at uniformly distributed $\varphi$. Then 
\begin{align*}
\tilde{r}^2&=|\tilde{z}-\alpha|^2
=|z(\varphi)\ee^{-\ii \Delta \psi}-\alpha|^2\\
&=|\alpha\left(\ee^{-\ii\Delta \psi}-1\right)+r\ee^{\ii \varphi}\ee^{-\ii\Delta \psi}|^2
\end{align*}
and 
\begin{align*}
\langle\tilde{r}^2\rangle&=2|\alpha|^2\left(1-\cos(\Delta\psi)\right)+r^2\,,\\
\langle\tilde{r}^2\rangle&=r^2+|\alpha|^2\langle\left(\Delta\psi\right)^2\rangle=r^2+|\alpha|^2\gamma^2\Delta t
\end{align*}
upon averaging in $\varphi$ first and then, for small $\Delta t$, in the jumps as well.

In conclusion: At large times $T$ Eq. (\ref{eq_N_T_infty_n_c}) is superseded by (cf.\ (\ref{eq_N_T_infty_c}))
\begin{equation}\label{eq_N_T_approx}
\langle \overline{N}_T\rangle\approx |\alpha|^2\gamma^2\frac{T}{2}\,
\end{equation}
because the average of $t$ is $T/2$. However for times $T\ll |\alpha|^{-2}\gamma^{-2}$
the observable $\overline{N}_T$ remains a good estimator of the quantum number $n$. The pointer observable $\mathcal{N}_t$ associated to (\ref{eq_pointer}) will then also serve its purpose, at least after some settling time that will be seen to be $T\gg \gamma^{-2}$.

We close the introduction by making a few selected references to the literature on the measurement process. Early on von Neumann \cite{VN96} mathematically discussed the Heisenberg cut, by which the conceptual dividing line between the observer and the observed system is meant. He showed that it can be shifted, providing examples which nowadays would be said to correspond to non-demolition measurements. Perez and Rosen \cite{P64} emphasized decoherence and the essential macroscopic character of the measurement apparatus, to which Hepp \cite{H72} supplied structure and examples. Merkli et al. \cite{BMS07, BMS08} rigorously exhibited decoherence for more general, not explicitly solvable examples. Fr\"ohlich \cite{BFS16} and coworkers give a notion of events that forgoes decoherence. From a more applied perspective, Makhlin \cite{MSS01} et al. discuss the limitations to an accurate readout of superconducting qubits. They are of a similar nature than in the present work.  

%
\end{section}

\begin{section}{Definition of the model and results}\label{section_model_results}
The precise definition of the model will be given in terms of its propagator, rather than the Hamiltonian, and in fact by means of quantum stochastic differential equation (QSDE).

Quantum stochastic calculus is a mathematical theory for quantum noise in open quantum systems developed by Hudson and Parthasarathy in 1984 \cite{PH84}. Applications are wide ranging. The first ones, which provided the construction of unitary dilations of quantum dynamical semigroups (\cite{PH84} already, and \cite{F85b}), were followed by other ones in quantum measurement theory \cite{BL85,B86}, and in quantum optics. The starting point of the latter applications is to model a Bosonic field using creation and annihilation processes. The first introduction of quantum stochastic calculus in this area was in \cite{GC85}, but related noise models already showed up earlier \cite{L66}. One kind of application of quantum stochastic calculus in quantum optics is in order to establish the master equation for the system of interest (e.g.\ \cite{G86}), another to describe the detection of photons \cite{B86,B87,B90,BMK03,BGM04,M87}, and yet another to model quantum input and output channels \cite{A88,B86,B88,B90,CW88,GC85,LRW88}. The latter is the field where we would locate our result. Other applications can for example be found in quantum filtering (e.g.\ \cite{B92a,B92b}).

Indirect measurements of quantum systems on the other hand have been discussed, including the non-demolition case, for different setups. The focus and the methods applied however differ from those presented in this article. We mention the experiment \cite{BHLRZ90, GBDSGKBRH07} by the Haroche group and the theoretical work of Bauer and Bernard \cite{BB11}, as well as \cite{BJM06, FF16, FF17a, FF17b, BMS07, BMS08}. In the experiment by Haroche, photons in a cavity are counted by letting them repeatedly interact with Rubidium atoms in a circular Rydberg state, which subsequently are measured. A description of such repeated indirect quantum non-demolition measurements was given by Bauer and Bernard.

The system proper we consider here is a harmonic oscillator. Its degrees of freedom are a creation and annihilation operator, $a^*$ and $a$, with commutation relation $[a,a^*]=1$ irreducibly realized on some Hilbert space $\mathcal{H}$. The operators 
\begin{equation}\label{eq_def_H_Gamma}
H=\omega\left(a^*a-\left(\overline{\alpha}a+\alpha a^*\right)\right)\,,\quad \Gamma=\gamma a^*a
\end{equation}
($\omega,\gamma >0$, $\alpha\in \mathbb{C}$) represent the Hamiltonian and the coupling strength to the measurement apparatus. Because the two operators do not commute, the measurement will not be of the non-demolition type, as pointed out earlier.

The state space of the apparatus is the bosonic Fock space $\mathcal{F}=\mathcal{F}(L^2(\mathbb{R}_+))$, with vacuum state $\ket{\Omega}$ and with creation and annihilation operators $A^*(f)$, $A(g)$, 
\begin{equation*}
[A(g),A^*(f)]=\braket{g}{f}1\,,\quad (f,g\in L^2(\mathbb{R}_+))\,.
\end{equation*}
We set in particular $A_t=A(1_{[0,t]})$, $(t\geq 0)$, and 
\begin{equation*}
Q_t=A_t+A_t^*\,,\quad P_t=-\ii(A_t-A_t^*)\,,
\end{equation*}
resulting in 
\begin{equation}\label{eq_comm_rel_QP}
\ii[P_t,Q_t]=2t\,.
\end{equation}
Finally the propagator is a solution of the QSDE
\begin{equation}\label{eq_U_t}
\dd U_t=-\left(\ii H+\frac{1}{2}\Gamma^2\right)U_t\dd t-\ii \Gamma U_t\dd P_t\,,\quad U_0=\mathbb{1}\,,
\end{equation}
which is of so-called Hudson-Parthasarathy form. The last term reflects the interaction (\ref{eq_H_I}), except that the equation is written in the It\^o rather than in the Stratonovich sense. Moreover, the equation is of a form that ensures the unitarity of $U_t$, or at least would if the operators were bounded. This point of precision is however inconsequential.

The rescaling of time ($t \mapsto \varepsilon t$) on the quantum field is represented by a unitary operator $T_\varepsilon$ acting as the second quantization of $f\mapsto f_\varepsilon$, $f_\varepsilon(t) = \sqrt{\varepsilon} f(\varepsilon t)$, i.e.
\begin{equation*}
T_\varepsilon A(f) T_\varepsilon^* = A (f_\varepsilon), \quad T_\varepsilon\ket{\Omega}=\ket{\Omega}.
\end{equation*}
In particular, 
\begin{equation}\label{eq:tscaling}
T_\varepsilon P_{\varepsilon t} T_\varepsilon^* = \sqrt{\varepsilon} P_t,
\end{equation}
and $T_\varepsilon \dd P_{\varepsilon t} T_\varepsilon^* = \sqrt{\varepsilon} \dd P_t$. This reflects the Wiener scaling, by which the probability distributions of the rescaled Brownian motion $W_{\varepsilon t}$  and  $\sqrt{\varepsilon} W_t$ are equal. For the evolution of the system, this implies that if $U(t)$ is the solution of (\ref{eq_U_t}) for parameters $(\omega, \alpha, \gamma)$ then 
\begin{equation*}\label{eq:Uscaling}
\tilde{U}(t) = T_\varepsilon U(\varepsilon t) T_\varepsilon^*
\end{equation*}
is the solution of the equation for $(\varepsilon \omega, \alpha, \sqrt{\varepsilon} \gamma)$. We conclude that the model has two dimensionless scales, $\alpha$ and $\omega \gamma^{-2}$.

The propagator can be computed explicitly. To this end we introduce the Weyl operators in the form
\begin{equation}\label{eq_dop}
D(z)= \ee^{z a^*-\overline{z} a}\,,\quad (z\in\mathbb{C})\,.
\end{equation}
They form a unitary projective representation of $\mathbb{C}$ by
\begin{equation*}
D(z')D(z)=\ee^{-\ii\im(\bar{z}'z)}D(z'+z)\,,\quad D(z)^{-1}=D(-z)=D(z)^*\,.
\end{equation*}

\begin{proposition}\label{prop_sol}
The family of operators on $\mathcal{H}\otimes \mathcal{F}$
\begin{equation}\label{eq_def_U_t}
U_t=\ee^{-\ii \phi_t a^*a}D(\mathcal{Z}_t)\ee^{-\ii G_t}\,,\quad (t\geq 0)
\end{equation}
is a solution to (\ref{eq_U_t}). Here 
\begin{gather}
\phi_t=\omega t+\gamma P_t\,,\nonumber\\
\mathcal{Z}_t=\ii \omega \alpha Z_t\,,\quad Z_t=\int_0^t\ee^{\ii \phi_s}\dd s\,,\label{eq_defs_U_t}\\
G_t=\omega^2|\alpha|^2\int_0^t\int_0^{s_1}\im \ee^{\ii\left(\phi_{s_2}-\phi_{s_1}\right)}\dd s_2 \dd s_1\,.\nonumber
\end{gather}
Moreover, $U_t$ is manifestly unitary for all $t$.
\end{proposition}
\begin{remark}\label{remark_Z_t_G_t}
$Z_t$ and $G_t$ are well-defined operators on $\mathcal{F}$, because the family $(P_s)_{0\leq s\leq t}$ is commuting; so is $D(\mathcal{Z}_t)$ because the exponent in (\ref{eq_dop}) remains self-adjoint up to a factor $\ii$.
\end{remark}
\begin{remark}
The field $P_t$ may be identified with a random process, namely Brownian motion $W_t$, see above and in more detail in Sect.~\ref{pf_limiting}. By (\ref{eq_defs_U_t}) we have
\begin{equation}\label{eq_pf_Z_t}
Z_t=Z_s+\ee^{\ii\phi_s}\int_s^t\ee^{\ii(\phi_r-\phi_s)}\dd r\equiv Z_s+\ee^{\ii \phi_s}\hat{Z}_{t-s}\,,
\end{equation}
where, for fixed $s$, $(\hat{Z}_{t-s})_{t\geq s}$ is a process independent of $(\ee^{\ii \phi_t})_{0\leq t\leq s}$ (and hence of $Z_s$), but equal in distribution, $(\hat{Z}_{t-s})_{t\geq s}\overset{d}{=}(Z_{\tau})_{\tau\geq 0}$ with $\tau=t-s$. In the limit $\varepsilon \to 0$, the phase factor $e^{\ii \phi_s}$ oscillates quickly and the process $Z_t$ has independent increments. Under the appropriate rescaling (\ref{eq:tscaling}), $Z_t$ becomes proportional to complex Brownian motion.

We also note that $R_t:=\ee^{-\ii\phi_t}Z_t$ satisfies the renewal equation $R_t=\ee^{-\ii(\phi_t-\phi_s)}R_s+\hat{R}_{t-s}$. (A related, but different renewal process is treated in \cite{GDL}.)  
\end{remark}

The next result says that $Q_t$ indeed meters the excitation number $a^*a$:
\begin{proposition}\label{prop_dU} 
We have
\begin{equation}\label{eq_prop_U_t_astara}
  \dd_s\left(U_s^*Q_tU_s\right)=
  \begin{cases}2\gamma U_s^*a^*aU_s\dd s\,, &\quad (0\leq s< t)\,,\\
    0\,, &\quad (s\ge t)\,.
    \end{cases}
\end{equation}
\end{proposition}

The increments then add up as follows:

\begin{proposition}\label{prop_HEQ}
\begin{gather}
U_t^*a^*aU_t=\left(a^*+\mathcal{Z}_t^*\right)\left(a+\mathcal{Z}_t\right)\,,\label{eq_prop_Q_t_ev_1}\\
U_t^*Q_tU_t=X_{2,t}a^*a+X_{1,t}^*a+X_{1,t}a^*+X_{0,t}\label{eq_prop_Q_t_ev_2}\,,
\end{gather}
where 
\begin{gather}
X_{2,t}=2\gamma t\,,\label{eq_prop_Q_t_ev_3}\\
X_{1,t}=2\ii \gamma \omega \alpha Y_{1,t}\,,\quad Y_{1,t}=\int_0^tZ_s\dd s\,,\label{eq_prop_Q_t_ev_4}\\ 
X_{0,t}=Q_t+2\gamma \omega^2 |\alpha|^2Y_{0,t}\,,\quad Y_{0,t}=\int_0^tZ_s^*Z_s\dd s\,.\label{eq_prop_Q_t_ev_5}
\end{gather}
\end{proposition}

The next result computes the expectation and the variance of the two observables. More precisely, let us focus on the initial state $\ket{n,\Omega}$, where $\ket{n}\in \mathcal{H}$, ($n\in\mathbb{N}$), is the eigenstate of the excitation number, $a^*a\ket{n}=n\ket{n}$, and $\ket{\Omega}\in\mathcal{F}$ is the field vacuum, $A_t\ket{\Omega}=0$, ($t\geq 0$). We set
\begin{equation*}
\langle A \rangle=\BraKet{n,\Omega}{A}{n,\Omega}\,,\quad \llangle A^2 \rrangle =\langle A^2\rangle-\langle A\rangle^2
\end{equation*}
for any operator $A$ on $\mathcal{H}\otimes \mathcal{F}$.

\begin{proposition}\label{prop_exp_var}
We have the following expectations and variances in the state $\ket{n,\Omega}$:
\begin{align}
&\E{U_t^*a^*aU_t}=n+\omega^2|\alpha|^2\E{Z_t^*Z_t}\,,\label{eq_prop_exp_var_1}\\
&\EE{(U_t^*a^*aU_t)^2}=\omega^2|\alpha|^2\left(\omega^2|\alpha|^2\EE{(Z_t^*Z_t)^2}+(2n+1)\E{Z_t^*Z_t}\right)\,,\label{eq_prop_exp_var_2}\\
&\E{U_t^*Q_tU_t}=2\gamma \left(t n+\omega^2|\alpha|^2\E{Y_{0,t}}\right)\,,\label{eq_prop_exp_var_3}\\
&\EE{(U_t^*Q_tU_t)^2}=t+4\gamma^2\omega^2|\alpha|^2\left(\omega^2|\alpha|^2\EE{Y_{0,t}^2}+(2n+1)\E{Y_{1,t}^*Y_{1,t}}\right)\,.\label{eq_prop_exp_var_4}
\end{align}
\end{proposition}

Let us stress once more that (\ref{eq_prop_Q_t_ev_1}) is the instantaneous excitation number, whereas (\ref{eq_prop_Q_t_ev_2}), divided by $2\gamma t$, is its time-averaged value, as sampled by the apparatus initialized in $\ket{\Omega}$, cf.\ (\ref{eq_prop_U_t_astara}) and $\BraKet{\Omega}{Q_t}{\Omega}=0$. We shall thus focus on the observables
\begin{equation}\label{eq_def_A_B}
N_t=U_t^*a^*aU_t\,,\quad \mathcal{N}_t=\frac{U_t^*Q_tU_t}{2\gamma t}\,.
\end{equation}
We note in passing that the family $(\mathcal{N}_t)_{t\ge 0}$ is commuting, as appropriate to pointer observables. This follows because of
\begin{equation*}\label{eq_fldout}
U_s^*Q_tU_s=U_t^*Q_tU_t\,,\quad (s\ge t)
\end{equation*}
by (\ref{eq_prop_U_t_astara}) and because $(Q_t)_{t\ge 0}$ is commuting.

The quantum number $n=0,1,2,\dots$ labels the states $\ket{n,\Omega}$. Their values are spaced by one, and so are the expectations of $\E{N_t}$ and $\E{\mathcal{N}_t}$, cf.\ (\ref{eq_prop_exp_var_1}, \ref{eq_prop_exp_var_3}). The issue however is as to whether the measurement of $N_t$ and, more importantly, $\mathcal{N}_t$ can be used to reliably tell apart finitely many states $n$. To this end (a) their expectations ought to remain close to $n$, and (b) their variances ought to be small w.r.t.\ unity. These two conditions will require that $\alpha$ is small and determine a time interval during which they are met for $N_t$; as for $\mathcal{N}_t$, the interval should not be too short, resulting in a further condition on $t$, $\gamma$. It ensures that the initial uncertainty of the pointer of the apparatus has been effaced. A sketch of how this conclusions are reached is as follows, with details supplied later in Sect.~\ref{section_proof_bds}.

\begin{enumerate}[leftmargin=1.9em, label=(\alph*)]
\item \label{condition_a}Expectations:
\begin{itemize}
\item Bounds
\begin{gather}
0\leq \E{N_t}-n\leq |\alpha|^2(4+\gamma^2 t)\,,\label{eq_exp_bd_1}\\
0\leq \E{\mathcal{N}_t}-n\leq |\alpha|^2(4+\frac{\gamma^2 t}{2})\,.\label{eq_exp_bd_2}
\end{gather}
\item Asymptotics in the regime $\gamma^2\ll \omega$, cf.\ (\ref{eq_wcr}):
\begin{align}
\E{N_t}-n\cong |\alpha|^2&\begin{cases} (\omega t)^2\,,\quad (t\ll \omega^{-1})\,,\\
2(1-\cos \omega t)\,,\quad (\omega^{-1}\ll t\ll \gamma^{-2})\,,\\
\gamma^2t\,,\quad (t\gg \gamma^{-2})\,,
\end{cases}\label{eq_exp_mathcalA}\\
\E{\mathcal{N}_t}-n\cong |\alpha|^2&\begin{cases} \frac{(\omega t)^2}{3}\,,\quad (t\ll \omega^{-1})\,,\\
2\,,\quad (\omega^{-1}\ll t\ll \gamma^{-2})\,,\\
\frac{\gamma^2t}{2}\,,\quad (t\gg \gamma^{-2})\,.
\end{cases}\label{eq_exp_mathcalB}
\end{align}
The last two cases match what was expected in (\ref{eq_N_T_infty_n_c}, \ref{eq_N_T_approx}). 
\end{itemize}
\end{enumerate}

The requirements set by condition (a), i.e. that the right hand sides of eqs.~(\ref{eq_exp_bd_1})-(\ref{eq_exp_mathcalB}) are $\ll 1$, are therefore 
\begin{equation}\label{eq_time_int_cond_a}
|\alpha|\ll 1\,,\quad t\ll |\alpha|^{-2}\gamma^{-2}
\end{equation}
for both observables.

\begin{enumerate}[leftmargin=1.9em, label=(\alph*)]
\setcounter{enumi}{1}
\item \label{condition_b}Variances:
\begin{itemize}
\item $N_t$: Both terms on the r.h.s.\ of (\ref{eq_prop_exp_var_2}) ought to be $\ll 1$. For both of them, this condition does not further limit the time interval (\ref{eq_time_int_cond_a}). For the second term this is even manifest, since it equals the one discussed in connection with (\ref{eq_prop_exp_var_1}).
\item $\mathcal{N}_t$: In view of the normalization of $\mathcal{N}_t$, the terms on the r.h.s.\ of (\ref{eq_prop_exp_var_4}) ought to be $\ll (\gamma t)^2$. For the last two terms, this yields the upper bounds (\ref{eq_time_int_cond_a}). The first term however sets the lower bound 
\begin{equation}\label{eq_time_int_cond_b2}
t\gg \gamma^{-2}\,.
\end{equation}
\end{itemize}
\end{enumerate}

\noindent
{\bf Summary.} The window of opportunity for the effective measurement of $n$ through $N_t$ is set by (\ref{eq_time_int_cond_a}); that for $\mathcal{N}_t$ is further restricted by (\ref{eq_time_int_cond_b2}).\\

The following result estimates the accuracy of $\mathcal{N}_t$ as an estimator for $N=a^*a$ for arbitrary initial states $\ket{\psi}$ of the oscillator.

\begin{proposition}\label{prop_estimator_process}
For any normalized $\ket{\psi}\in \mathcal{H}$ we have
\begin{equation*}
\BraKet{\psi,\Omega}{(\mathcal{N}_t-a^*a)^2}{\psi,\Omega}\leq C_1|\alpha|^2\bigl(1+\gamma^2 t\bigr)\BraKet{\psi}{2a^*a+1}{\psi}+C_2|\alpha|^4\Bigl(1+\bigl(\gamma^2 t\bigr)^2\Bigr)+\bigl(\gamma^2t\bigr)^{-1}\,.
\end{equation*}
In particular the r.h.s.\ is small for (\ref{eq_time_int_cond_a}, \ref{eq_time_int_cond_b2}), provided the excitation number is bounded.
\end{proposition}

We conclude this section by discussing the long time limiting regime of the dynamics. To this end we introduce a notion of convergence that allows to express the stochastic but classical nature of the limiting dynamics. It comes in three variants (i-iii). First, (i) we say that the self-adjoint  operators $X_{\varepsilon}$ on $\mathcal{F}$ tend to the $\mathbb{R}$-valued random variable $X$ on some probability space $(\tilde{\Omega}, \tilde{\mu})$ w.r.t. the vacuum $\ket{\Omega}\in \mathcal{F}$ as $\varepsilon\to 0$ (written as $X_{\varepsilon}\to_{\ket{\Omega}}X $, ($\varepsilon\to 0$)) if 
\begin{equation}\label{eq_def_conv_ltl}
\BraKet{\Omega}{f(X_{\varepsilon})}{\Omega}\to \int_{\tilde{\Omega}} f(X(\omega))\dd \tilde{\mu}(\omega)\,,\quad (\varepsilon\to 0)
\end{equation}
for every bounded continuous function $f:\mathbb{R}\to\mathbb{C}$. Second, let $C(\mathbb{R}_+,\mathbb{R})$ denote the classical Wiener space of continuous functions $\mathbb{R}_+\to \mathbb{R}$. Then (ii) is an extension to commuting families $(X_{\varepsilon,t})_{t\ge 0}$ of self-adjoint operators on the l.h.s. of (\ref{eq_def_conv_ltl}) and to stochastic processes $(X_t)_{t\ge 0}$, i.e. random variables taking values in $C(\mathbb{R}_+,\mathbb{R})$, on the r.h.s.; the functions $f$ are now $C(\mathbb{R}_+,\mathbb{R})\to\mathbb{C}$. Finally, (iii), both variants can be extended by replacing $\mathcal{F}$ with $\mathcal{H}\otimes\mathcal{F}$. The random variable $X$ and the stochastic processes $(X_t)_{t\ge 0}$, respectively, are then to be understood as multiples of the identity operator on $\mathcal{H}$; the convergence is meant in the sense of weak convergence of bounded operators on $\mathcal{H}$. Explicitly, (\ref{eq_def_conv_ltl}) is then to be read as 
\begin{equation}\label{eq_def_conv_ltl_expl}
\BraKet{\psi,\Omega}{f(X_{\varepsilon})}{\psi,\Omega}\to \int_{\tilde{\Omega}} f(X(\omega))\dd \tilde{\mu}(\omega)\braket{\psi}{\psi}\,,\quad (\varepsilon\to 0)
\end{equation}
which expresses that the limit is (a) oblivious to the quantum state $\ket{\psi}\in\mathcal{H}$, ($\Vert \psi\Vert=1$) of the oscillator, and (b) indeed given by a classical random process.

In view of the linearly growing expectation of $N_t$ at large times $t$, cf.~(\ref{eq_exp_mathcalA}), a non-trivial scaling limit ought to be given by $\varepsilon N_{t/\varepsilon}$, ($\varepsilon\to 0$), as the following result confirms. 
\begin{proposition}[Long time limit]\label{prop_ltl}
Let $\kappa=\omega \alpha\gamma /c$, where
\begin{equation}\label{eq:c}
  c=\ii \omega-\gamma^2/2\,.
\end{equation}
Then 
\begin{gather}
\varepsilon \left.N_t\right|_{t\to \varepsilon^{-1}t}\to_{\ket{\Omega}} |\kappa|^2|B_t|^2\,,\label{ltl_N}\\
\varepsilon \left.\mathcal{N}_t\right|_{t\to \varepsilon^{-1}t}\to_{\ket{\Omega}} |\kappa|^2\frac{1}{t}\int_0^t|B_s|^2\dd s\,,\label{ltl_mcN}
\end{gather}
as $\varepsilon\to 0$, where $B_t=\re B_t+\ii \im B_t$ is complex Brownian motion. Its components are independent and distributed as $W_t/\sqrt{2}$, where $W_t$ is real Brownian motion. (This normalization ensures that $|\dd B_t|^2=\dd t$, as in the real case.) The two limits are in the sense (i) and (ii), respectively, as extended in (iii).  The limits hold for $t>0$ and for $t\in I$, respectively, where $I\subset \mathbb{R}_+$ is any compact interval.
\end{proposition}
In plain terms the proposition states that at large times no trace is
left of the initial state of the oscillator  and that the excitation number follows the square displacement of a diffusive motion of constant $|\kappa|^2/4$. 
\begin{remark}
\begin{enumerate}
\item The different kinds of limit are due to the fact that only the operators seen on the l.h.s. of (\ref{ltl_mcN}) form a commuting family. 
\item The characteristic function of $|W_t|^2$ (Laplace transform of its measure) is $(1+2\lambda t)^{-1/2}$, that of $t^{-1}\int_0^t |W_s|^2\dd s$ is $(\cosh\sqrt{2\lambda t})^{-1/2}$. The first is immediate, the second was derived in \cite{K49}, where that of $\int_0^t|W_s|^2\dd s$ was found to be $(\cosh\sqrt{2\lambda} t)^{-1/2}$. The characteristic function of the complex counterparts is then found by the replacements $\chi(\lambda)\to \chi(\lambda/2)^2$ to be $(1+\lambda t)^{-1}$ and $(\cosh \sqrt{\lambda} t)^{-1}$, respectively.
\end{enumerate}
\end{remark}

\end{section}

\begin{section}{Proofs of identities}

\begin{proof}[Proof of Prop. \ref{prop_sol}]
We first observe that, according to the claim (\ref{eq_defs_U_t}), $\mathcal{Z}_t$ and $G_t$ are differentiable, whereas $\phi_t$ is a stochastic integral. We shall thus regard (\ref{eq_def_U_t}) as a generic ansatz of that type for a solution of (\ref{eq_U_t}). Three preliminaries are in order. The first one is 
\begin{equation*}
\frac{\dd }{\dd t}D(z_t)=\left(\dot{z}_ta^*-\dot{\overline{z}}_t a+\ii \im(\dot{\overline{z}}_tz_t)\right)D(z_t)\,,
\end{equation*}
where $z_t$ takes values in $\mathbb{C}$ as in (\ref{eq_dop}) or is promoted to a differentiable multiplication operator in the variables $(P_s)_{0\leq s\leq t}$, as in Remark \ref{remark_Z_t_G_t}. That claim follows by differentiating the equation
$$
D(z_s) = \ee^{\ii \im(\overline{z}_sz_t)}D(z_s-z_t) D(z_t)
$$
by $s$ at the point $s=t$ and using
$$
\left.\frac{\dd }{\dd s}D(z_s-z_t)\right|_{s=t}=\dot{z}_ta^*-\dot{\overline{z}}_t a\,.
$$
The second preliminary is 
\begin{equation*}
\ee^{-\ii\phi_ta^*a}\left(\dot{z}_ta^*-\overline{z}_ta\right) \ee^{\ii\phi_ta^*a}=\dot{z}_t\ee^{-\ii\phi_t}a^*-\overline{z}_t\ee^{\ii\phi_t}a\,,
\end{equation*}
while the last one is
\begin{equation*}
\dd \left(\ee^{-\ii \phi_t a^*a}\right)=\Bigl(-\ii a^*a\dd \phi_t-\frac{1}{2}(a^*a)^2\left(\dd \phi_t\right)^2\Bigr)\ee^{-\ii\phi_t a^*a}
\end{equation*}
and follows from It\^o's lemma \cite{PH84}. So prepared we differentiate (\ref{eq_def_U_t}) and obtain
\begin{equation*}
\dd U_t=\Bigl(-\ii a^*a \dd \phi_t-\frac{1}{2}(a^*a)^2\left(\dd \phi_t\right)^2+\bigl( a^*\dot{\mathcal{Z}}_t\ee^{-\ii \phi_t}- a\dot{\mathcal{Z}}_t^*\ee^{\ii\phi_t}+\ii \im(\dot{\mathcal{Z}}_t^*\mathcal{Z}_t)-\ii \dot{G}_t\bigr)\dd t\Bigr)U_t\,.
\end{equation*}
Comparing coefficients of $1$, $a^*$ (or $a$), $a^*a$ and $(a^*a)^2$ with those in (\ref{eq_def_H_Gamma}, \ref{eq_U_t}) we obtain
\begin{gather*}
\im\left(\dot{\mathcal{Z}}_t^*\mathcal{Z}_t\right)-\dot{G}_t=0\,,\\
\dot{\mathcal{Z}}_t\ee^{-\ii \phi_t}=\ii \omega \alpha\,,\\
-\ii\dd\phi_t=-\ii\omega\dd t-\ii\gamma\dd P_t\,,\\
-\frac{1}{2}\left(\dd \phi_t\right)^2=-\frac{\gamma^2}{2}\dd t\,,
\end{gather*}
as well as $\phi_0=0$, $\mathcal{Z}_0=0$, $G_0=0$. The last of the four equations is a consequence of the third, which together with the second is solved by the expressions (\ref{eq_defs_U_t}) for $\phi_t$ and $\mathcal{Z}_t$; and so is the first one by that for $G_t$, once it is restated as
\begin{equation*}
\dot{G}_t=\omega^2|\alpha|^2\int_0^t\im\ee^{\ii(\phi_s-\phi_t)}\dd s\,.
\end{equation*}
\end{proof}
\begin{proof}[Proof of Prop. \ref{prop_dU}]
By the It\^o rule $\dd (MM')=(\dd M)M'+M(\dd M')+(\dd M)(\dd M')$ we obtain 
\begin{multline*}
\dd_s\left(U_s^*Q_tU_s\right)=\\U_s^*\bigl(\left(\ii[H,Q_t]-\left\{\Gamma^2/2,Q_t\right\}\right)\dd s+\ii\left(\Gamma(\dd P_s)Q_t-Q_t(\dd P_s)\Gamma\right)
+\Gamma(\dd P_s)Q_t(\dd P_s)\Gamma\bigr) U_s\,.
\end{multline*}
Let $0\le s< t$. In view of $[H,Q_t]=[\Gamma,Q_t]=[\Gamma,P_s]=0$ and of $[\dd P_s,Q_t]=-2\ii \dd s$, $(\dd P_s)^2=\dd s$, $(\dd s)(\dd P_s)=0$ we get
\begin{equation*}
-\left\{\Gamma^2/2,Q_t\right\}\dd s +\Gamma(\dd P_s)Q_t(\dd P_s)\Gamma=0
\end{equation*}
and thus
\begin{equation*}
\dd_s \left(U_s^*Q_tU_s\right)=2U_s^*\Gamma U_s\dd s\,,
\end{equation*}
as claimed. For $s\ge t$ the only change is by $[\dd P_s,Q_t]=0$.
\end{proof}
\begin{proof}[Proof of Prop. \ref{prop_HEQ}]
We have $U_t^*a^*aU_t=D(\mathcal{Z}_t)^*a^*aD(\mathcal{Z}_t)$ by (\ref{eq_def_U_t}) and 
\begin{equation*}
D(\mathcal{Z}_t)^*aD(\mathcal{Z}_t)=a+\mathcal{Z}_t
\end{equation*}
because $[a,\mathcal{Z}_ta^*-\mathcal{Z}_t^*a]=\mathcal{Z}_t$, whence Eq. (\ref{eq_prop_Q_t_ev_1}) follows. Furthermore, expanding the brackets in (\ref{eq_prop_Q_t_ev_1}) and integrating (\ref{eq_prop_U_t_astara}) yields Eq. (\ref{eq_prop_Q_t_ev_2}).
\end{proof}
\begin{proof}[Proof of Prop. \ref{prop_exp_var}]
Let us generically denote by $V$ any monomial $(a^*)^ma^n$ with $m\not=n$, as well as any linear combinations thereof. Then $\E{V}=0$. We have 
\begin{equation*}
U_t^*a^*aU_t=a^*a+\mathcal{Z}_t^*\mathcal{Z}_t+V\,,
\end{equation*}
with $V=a\mathcal{Z}_t^*+a^*\mathcal{Z}_t$, from which (\ref{eq_prop_exp_var_1}) follows, as well as 
\begin{equation*}
(U_t^*a^*aU_t)^2=\left(a^*a+\mathcal{Z}_t^*\mathcal{Z}_t\right)^2+\mathcal{Z}_t^*\mathcal{Z}_t(a^*a+aa^*)+V\,.
\end{equation*}
Thus
\begin{align*}
\EE{(U_t^*a^*aU_t)^2}&= \EE{\left(a^*a+\mathcal{Z}_t^*\mathcal{Z}_t\right)^2}+(2n+1)\E{\mathcal{Z}_t^*\mathcal{Z}_t}\\
&= \EE{(\mathcal{Z}_t^*\mathcal{Z}_t)^2}+(2n+1)\E{\mathcal{Z}_t^*\mathcal{Z}_t}\,,
\end{align*}
where in the last equality we used that $a^*a$ has no variance in $\ket{n}$, thus amounting to a shift of $\mathcal{Z}_t^*\mathcal{Z}_t$. So (\ref{eq_prop_exp_var_2}) follows by (\ref{eq_defs_U_t}).

The other observable is dealt with similarly. We have 
\begin{equation*}
U_t^*Q_tU_t=X_{2,t}a^*a+X_{0,t}+V\,,
\end{equation*}
from which (\ref{eq_prop_exp_var_3}) follows by (\ref{eq_prop_Q_t_ev_3}, \ref{eq_prop_Q_t_ev_5}) and $\E{Q_t}=0$, as well as
\begin{equation*}
(U_t^*Q_tU_t)^2=\left(X_{2,t}a^*a+X_{0,t}\right)^2+X_{1,t}^*X_{1,t}(a^*a+aa^*)+V\,.
\end{equation*}
Again, 
\begin{equation*}
\EE{(U_t^*Q_tU_t)^2}=\EE{X_{0,t}^2}+(2n+1)\E{X_{1,t}^*X_{1,t}}\,,
\end{equation*}
where we used that $X_{2,t}a^*a$ has no variance in $\ket{n,\Omega}$, cf.\ (\ref{eq_prop_Q_t_ev_3}). The last term on the r.h.s.\ yields the corresponding one in (\ref{eq_prop_exp_var_4}) by (\ref{eq_prop_Q_t_ev_4}). It remains to discuss the first term:
\begin{equation}\label{eq_pf_var_X_0}
\EE{X_{0,t}^2}=\E{P_t^2}+4\gamma^2 \omega^4 |\alpha|^4\EE{Y_{0,t}^2}\,.
\end{equation}
To see this, we observe that $Q_t+\ii P_t=2A_t$ and $A_t\ket{\Omega}=0$. Thus $X_{0,t}\ket{\Omega}=\tilde{X}_{0,t}\ket{\Omega}$, where $\tilde{X}_{0,t}$ is obtained from $X_{0,t}$ in (\ref{eq_prop_Q_t_ev_5}) by replacing $Q_t$ with $-\ii P_t$. Thus 
\begin{equation}\label{eq_pf_exp_X0t2}
\E{X_{0,t}^2}= \E{\tilde{X}_{0,t}^*\tilde{X}_{0,t}}= \E{P_t^2}+4\gamma^2 \omega^4 |\alpha|^4\E{ Y_{0,t}^2}
\end{equation}
and (\ref{eq_pf_var_X_0}) follows. Finally $\E{P_t^2}=t$.
\end{proof}
\end{section}

\begin{section}{Proofs of bounds}\label{section_proof_bds}
We shall supply the details leading to the conclusions given at the end of Sect.~\ref{section_model_results}. We observe that by the r.h.s.\ of (\ref{eq_prop_exp_var_1}-\ref{eq_prop_exp_var_4}) we are left with expectations and variances of operators $A$ on $\mathcal{F}$ only, i.e.\ $\E{A}=\BraKet{\Omega}{A}{\Omega}$. The expectations are computed as follows. 

\begin{lemma}\label{lemma_exp_FS}
Let $\mathcal{R}_n(f)(t)$ (also written $\mathcal{R}_n(f(t))$ with slight abuse of notation) be the remainder of the $n$-th Taylor expansion in $t=0$. We set $c=\ii \omega-\gamma^2/2$, as in (\ref{eq:c}). Then 
\begin{gather}
\E{\ee^{\ii\phi_t}}=\ee^{ct}\,,\label{eq_lem_exp_var_0}\\
\langle Z_t\rangle= c^{-1}\mathcal{R}_0(\ee^{ct})\,,\label{eq_lem_exp_var_1}\\
\langle Z_t^*Z_t\rangle= c^{-2}\mathcal{R}_1(\ee^{ct})+c.c.\,,\label{eq_lem_exp_var_2}\\
\langle Y_{1,t}\rangle=c^{-2}\mathcal{R}_1(\ee^{ct})\,,\label{eq_lem_exp_var_3}\\
\langle Y_{0,t}\rangle=  c^{-3}\mathcal{R}_2(\ee^{ct})+c.c.\,,\label{eq_lem_exp_var_4}\\
\langle Y_{1,t}^*Y_{1,t}\rangle=  c^{-4}\mathcal{R}_3((ct-1)\ee^{ct})+c.c.\,.\label{eq_lem_exp_var_5}
\end{gather}
\end{lemma}

The variances, on the other hand, are estimated as follows. We set
\begin{equation*}
\langle A ; B\rangle=\langle AB\rangle-\langle A\rangle \langle B\rangle\,,\quad \llangle A^2\rrangle=\langle A ; A\rangle\,.
\end{equation*}
\begin{lemma}\label{lemma_var_Z} 
\begin{gather}
\left|\langle  Z_t^*Z_t ; Z_t\rangle\right|\leq \frac{C\gamma^2}{\omega^3}t\,,\label{eq_lrangle_Z_t}\\
\EE{(Z_t^*Z_t)^2}\leq \frac{\gamma^2 }{\omega^4}(2\gamma^2t^2+Ct)\,,\label{eq_llangle_Z_t}
\end{gather}
and more generally 
\begin{equation}
\left|\langle Z_s^*Z_s ; Z_t^*Z_t\rangle\right|\leq \frac{\gamma^2 }{\omega^4}(2\gamma^2s^2+Cs)\label{eq_lrangle_Z_s_t}
\end{equation}
for $0\leq s\leq t$. In particular
\begin{equation}\label{eq_var_Y_0_t}
\llangle  Y_{0,t}^2\rrangle\leq \frac{\gamma^2 }{3\omega^4}\bigl(\gamma^2t^4+Ct^3\bigr)\,,
\end{equation}
where $C$ is a numerical constant changing from line to line.
\end{lemma}

We postpone the proofs of Lemma \ref{lemma_exp_FS} and \ref{lemma_var_Z} and continue towards the stated goal. The following estimates on Taylor remainders, valid for $\re z\leq 0$, will be used:
\begin{equation}\label{eq_est_Taylor_rem}
\left|\mathcal{R}_0(\ee^{z})\right|\leq 2\,,\quad \left|\frac{\mathcal{R}_1(\ee^{z})}{z}\right|\leq 2\,,\quad \left|\frac{\mathcal{R}_2(\ee^{z})}{z^2}\right|\leq 1\,.
\end{equation}
The first one is elementary, the second follows by the mean value theorem, and so does the third, yet by Cauchy's form. Also 
\begin{equation}\label{eq_id_c}
c^{-1}+{\overline{c}}^{-1}=-|c|^{-2}\gamma^2\,,\quad \omega \leq |c|
\end{equation}
will be often used.

We first focus on $N_t$ as defined in (\ref{eq_def_A_B}) and on the auxiliary quantities it calls for by (\ref{eq_prop_exp_var_1}, \ref{eq_prop_exp_var_2}). They will be written without the prefactors $|\alpha|^2$ or $|\alpha|^4$, which are to be included when finalizing the estimates.

\begin{itemize}[leftmargin=1.2em]
\item The first one is $\omega^2 \E{Z_t^*Z_t}$ and calls in turn for (\ref{eq_lem_exp_var_2}). We have $\mathcal{R}_1(\ee^{ct})=\mathcal{R}_0(\ee^{ct})-ct$ and hence 
\begin{equation}\label{eq_exp_Z*Z}
\E{Z_t^*Z_t}=c^{-2}\mathcal{R}_0(\ee^{ct})-c^{-1}t+c.c.\,.
\end{equation}
By (\ref{eq_est_Taylor_rem}, \ref{eq_id_c}) we find
\begin{equation*}
\omega^2 \E{Z_t^*Z_t}\leq 4+\gamma^2 t\,,
\end{equation*}
confirming (\ref{eq_exp_bd_1}). In the regime $\gamma^2\ll \omega$ the more detailed discussion goes as follows: 
\begin{enumerate}[label=\roman*)]
\item For $\omega t\ll 1$, or equivalently $|c|t\ll 1$, we have $\mathcal{R}_1(\ee^{ct})\cong (ct)^2/2$ and thus
\begin{equation*}
\omega^2\E{Z_t^*Z_t}\cong \frac{(\omega t)^2}{2}+c.c.\,.
\end{equation*}
\item For $\omega t\gg 1 \gg \gamma^2 t$ we use (\ref{eq_exp_Z*Z}), wherein $\ee^{ct}\cong \ee^{\ii \omega t}$ and $c^{-2}\mathcal{R}_0(\ee^{ct})\cong -\omega^2(\ee^{\ii\omega t}-1)$. We end up with (\ref{eq_exp_mathcalA}). 
\item For $ \gamma^2 t \gg 1$ we have $\E{Z_t^*Z_t} \cong -c^{-1} t + c.c. = \gamma^2 t$.
\end{enumerate}
\item The second auxiliary quantity is $\omega^4\EE{(Z_t^*Z_t)^2}$  and calls for (\ref{eq_llangle_Z_t}). For reasons stated below (\ref{eq_time_int_cond_a}), it yields the condition $|\alpha|^4M\ll 1$ with $M=\gamma^2t(2\gamma^2t+C)$. Given that $|\alpha|\ll 1$, this still is  (\ref{eq_time_int_cond_a}).
\end{itemize}

We next focus on $\mathcal{N}_t$ and on the auxiliary quantities it calls for by (\ref{eq_prop_exp_var_3}, \ref{eq_prop_exp_var_4}).

\begin{itemize}[leftmargin=1.2em]
\item The first one is 
\begin{equation}\label{eq_exp_Y0+Y0*}
\frac{\omega^2}{t}\E{Y_{0,t}}\,,
\end{equation}
and calls in turn for (\ref{eq_lem_exp_var_4}). We have $\mathcal{R}_2(\ee^{ct})=\mathcal{R}_1(\ee^{ct})-(ct)^2/2$ and hence 
\begin{equation*}
\frac{\E{Y_{0,t}}}{t}=c^{-2}\frac{\mathcal{R}_1(\ee^{ct})}{ct}-\frac{c^{-1}t}{2}+c.c.\,.
\end{equation*}
By (\ref{eq_prop_Q_t_ev_5}) and (\ref{eq_est_Taylor_rem}, \ref{eq_id_c}) we find
\begin{equation}\label{eq_exp_Y0tomega2}
0\leq \frac{\omega^2}{t}\E{Y_{0,t}}\leq 4+\frac{\gamma^2 t}{2}\,,
\end{equation}
confirming (\ref{eq_exp_bd_2}). In the regime $\gamma^2\ll \omega$, a more detailed discussion goes as follows:
\begin{enumerate}[label=\roman*)]
\item For $\omega t\ll 1$, or equivalently $|c|t\ll 1$, we have $\mathcal{R}_2(\ee^{ct})\cong (ct)^3/6$ and thus
\begin{equation*}
\frac{\omega^2}{t}\E{Y_{0,t}}\cong \frac{(\omega t)^2}{3}\,.
\end{equation*}
\item For $\omega t\gg 1$ we write 
\begin{equation*}
\mathcal{R}_2(\ee^{ct})=\mathcal{R}_0(\ee^{ct})-ct-(ct)^2/2
\end{equation*}
and estimate the contribution to (\ref{eq_exp_Y0+Y0*}) of the three terms separately. Since $|\mathcal{R}_0(\ee^{ct})|\leq 2$ we have 
\begin{equation*}
\frac{\omega^2}{t}\left|c^{-3}\mathcal{R}_0(\ee^{ct})\right|\leq \frac{2}{\omega t}\ll 2\,.
\end{equation*}
Moreover by $-c^{-2}+c.c.\cong 2\omega^{-2}$ we have 
\begin{equation*}
\frac{\omega^2}{t}c^{-3}(-ct+c.c.)\cong 2
\end{equation*}
and 
\begin{equation*}
\frac{\omega^2}{t}c^{-3}\left(\frac{-(ct)^2}{2}+c.c.\right)\cong \frac{\gamma^2 t }{2}\,.
\end{equation*}
\end{enumerate}
These findings are summarized in (\ref{eq_exp_mathcalB}).
\item The second auxiliary quantity is 
\begin{equation*}\label{eq_exp_Y1*Y1}
\frac{\omega^2}{t^2}\E{Y_{1,t}^*Y_{1,t}}\,,
\end{equation*}
and calls for (\ref{eq_lem_exp_var_5}). 
We have $\mathcal{R}_3((ct-1)\ee^{ct})=\mathcal{R}_2((ct-1)\ee^{ct})-(ct)^3/3$ and hence 
\begin{equation*}
\frac{\E{Y_{1,t}^*Y_{1,t}}}{t^2}=c^{-2}\frac{\mathcal{R}_2((ct-1)\ee^{ct})}{(ct)^2}-\frac{c^{-1}t}{3}+c.c.\,.
\end{equation*}
By (\ref{eq_est_Taylor_rem}, \ref{eq_id_c}) we find 
\begin{equation}\label{eq_exp_Y1tstarY1tomega2}
\frac{\omega^2}{t^2}\E{Y_{1,t}^*Y_{1,t}}\leq 2+\frac{\gamma^2t}{3}\,.
\end{equation}
For reasons stated before (\ref{eq_exp_bd_2}) that bound yields the condition
\begin{equation*}
|\alpha|^2\left(2+\frac{\gamma^2 t}{3}\right)\ll 1\,,
\end{equation*}
i.e.\ (\ref{eq_time_int_cond_a}), as announced before (\ref{eq_time_int_cond_b2}).
\item The third auxiliary quantity is 
\begin{equation*}
\frac{\omega^4 }{t^2}\EE{Y_{0,t}^2}
\end{equation*}
and calls for (\ref{eq_var_Y_0_t}). By the reasons just recalled it yields the condition $|\alpha|^4M\ll 1$ seen earlier and thus (\ref{eq_time_int_cond_a}) once more.
\end{itemize}

This completes the discussion of items \ref{condition_a}, \ref{condition_b} in Sect.~\ref{section_model_results}. 

\begin{proof}[Proof of Prop. \ref{prop_estimator_process}]
By (\ref{eq_prop_Q_t_ev_2}, \ref{eq_prop_Q_t_ev_3}), we have 
\begin{equation*}
(\mathcal{N}_t-a^*a)^2=\frac{1}{(2\gamma t)^2}\bigl(X_{1,t}^*a+X_{1,t}a^*+X_{0,t}\bigr)^2\leq \frac{3}{(2\gamma t)^2}\bigl(X_{1,t}^*X_{1,t}(2a^*a+1)+X_{0,t}^2\bigr)
\end{equation*}
by the Cauchy-Schwarz inequality. Expectations in $\ket{\Omega}$ are computed by (\ref{eq_prop_Q_t_ev_4}) as $\E{X_{1,t}^*X_{1,t}}=4\gamma^2\omega^2|\alpha|^2\E{Y_{1,t}^*Y_{1,t}}$, and by (\ref{eq_pf_exp_X0t2}). The proof will thus be finished by showing
\begin{gather*}
\omega^2\E{Y_{1,t}^*Y_{1,t}}\leq Ct^2(1+\gamma^2t)\,,\\
\omega^4\E{Y_{0,t}^2}\leq Ct^2(1+(\gamma^2t)^2)\,.
\end{gather*}
The first bound is just (\ref{eq_exp_Y1tstarY1tomega2}). The second one follows from $\E{Y_{0,t}^2}=\E{Y_{0,t}}^2+\EE{Y_{0,t}^2}$ together with (\ref{eq_exp_Y0tomega2}, \ref{eq_var_Y_0_t}) and $\gamma^2t^3\leq t^2(1+(\gamma^2t)^2)$.
\end{proof}

We next catch up on the proofs of Lemma \ref{lemma_exp_FS} and \ref{lemma_var_Z}. We begin by some preliminaries, which will be used repeatedly.

\begin{enumerate}[leftmargin=1.9em, label=\alph*)]
\item By iterating $f(t)=f(0)+\int_0^tf'(s)\dd s$ 
we get the integral form of the remainder of the $n$-th Taylor approximation ($n=0,1,2,\dots$)
\begin{equation*}
\mathcal{R}_n(f)(t)=\int_0^t\mathcal{R}_{n-1}(f')(s)\dd s\,
\end{equation*}
with $\mathcal{R}_{-1}(f)=f$; equivalently 
\begin{equation}\label{eq_remainder_Taylor}
\mathcal{R}_n(f)(0)=0\,,\quad \mathcal{R}_n(f)'=\mathcal{R}_{n-1}(f')\,.
\end{equation}
\item \label{prelim_equality_distr}We observe that $(\phi_t-\phi_s)_{0\leq s\leq t}\overset{d}{=}(\phi_{t-s})_{0\leq s \leq t}$, where $\overset{d}{=}$ means equality in distribution. Thus 
\begin{equation}\label{eq_Z_Z_star}
\ee^{-\ii \phi_t}Z_t=\int_0^t\ee^{-\ii(\phi_t-\phi_s)}\dd s \overset{d}{=}\int_0^t\ee^{-\ii\phi_s}\dd s=Z_t^*\,.
\end{equation}
\item A continuous function on an interval, which has a continuous right derivative, is differentiable (see e.g.\ \cite{P83}, Cor.\ 1.2). Claims on derivatives can thus be read as being taken from the right.\label{prelim_diff_cont}
\item Let $t\geq s$ and let $F_s$ be a functional of $(P_{\tau})_{0\leq \tau \leq s}$. Then by independence and (\ref{eq_lem_exp_var_0}),
\begin{equation*}
\E{F_s\ee^{\ii\phi_t-ct}}=\E{F_s\ee^{\ii\phi_s-cs}}\E{\ee^{\ii(\phi_t-\phi_s)}}\ee^{-c(t-s)}=\E{F_s\ee^{\ii\phi_s-cs}}\,.
\end{equation*}
In particular
\begin{equation}\label{eq_diff_A_s_exp_t}
\frac{\dd }{\dd t}\left.\E{F_s\ee^{\ii\phi_t-ct}}\right|_{t=s}=0\,.
\end{equation}
This last preliminary shall not be used in the proof of (\ref{eq_lem_exp_var_0}), since it depends on it.
\end{enumerate}
\begin{proof}[Proof of Lm. \ref{lemma_exp_FS}]
The basic expectation value is 
\begin{equation*}
\langle \ee^{\ii \gamma P_t}\rangle =\ee^{-\gamma^2t/2}\,,\quad (t\geq 0)\,,
\end{equation*}
in view of
$\ee^{\ii \gamma P_t}=\ee^{\gamma (A_t-A_t^*)}=\ee^{-\gamma A_t^*}\ee^{\gamma A_t}\ee^{-\gamma^2t/2 }$.
Thus 
\begin{equation*}
  \langle \ee^{\ii \phi_t}\rangle =\ee^{\ii \omega t}\langle \ee^{\ii \gamma P_t}\rangle =\ee^{ct}
\end{equation*}
for $t\geq 0$ proving (\ref{eq_lem_exp_var_0}). We then obtain 
\begin{equation*}
\frac{\dd }{\dd t}\E{Z_t}=\E{\ee^{\ii\phi_t}}=\ee^{ct}\,
\end{equation*}
which by (\ref{eq_remainder_Taylor}) implies (\ref{eq_lem_exp_var_1}). Next
\begin{equation*}
\frac{\dd }{\dd t}\E{Z_t^*Z_t}=\E{\ee^{-\ii\phi_t}Z_t}+\E{Z_t^*\ee^{\ii\phi_t}}=\E{Z_t^*}+\E{Z_t}\,
\end{equation*}
by (\ref{eq_Z_Z_star}), whence (\ref{eq_lem_exp_var_1}) implies (\ref{eq_lem_exp_var_2}). From 
\begin{equation*}
\frac{\dd }{\dd t}\E{Y_{1,t}}=\E{Z_t}\,,\quad \frac{\dd }{\dd t}\E{Y_{0,t}}=\E{Z_t^*Z_t}
\end{equation*}
the equations (\ref{eq_lem_exp_var_3}, \ref{eq_lem_exp_var_4}) follow, too. Finally 
\begin{equation*}
\frac{\dd }{\dd t}\E{Y_{1,t}^*Y_{1,t}}=\E{Z_t^*Y_{1,t}}+\E{Y_{1,t}^*Z_t}\,,\quad \frac{\dd }{\dd t}\E{Y_{1,t}^*Z_t}=\E{Z_t^*Z_t}+\E{Y_{1,t}^*\ee^{\ii\phi_t}}\,,
\end{equation*}
and, by (\ref{eq_diff_A_s_exp_t}, \ref{eq_lem_exp_var_1})
\begin{equation*}
\frac{\dd }{\dd t}\E{Y_{1,t}^*\ee^{\ii\phi_t-ct}}=\E{Z_t^*\ee^{\ii\phi_t-ct}}=\ee^{-ct}\E{Z_t}=c^{-1}\bigl(1-\ee^{-ct}\bigr)\,.
\end{equation*}
Hence 
\begin{equation*}
\E{Y_{1,t}^*\ee^{\ii\phi_t}}=\ee^{ct}\bigl(c^{-1}t+c^{-2}\bigl(\ee^{-ct}-1\bigr)\bigr)=c^{-2}\bigl((ct-1)\ee^{ct}+1\bigr)=c^{-2}\mathcal{R}_1\bigl((ct-1)\ee^{ct}\bigr)
\end{equation*}
since $g(x)=(x-1)\ee^{x}$ has $g(0)=-1$, $g'(0)=0$. We conclude 
\begin{align*}
\frac{\dd^2 }{\dd t^2}\E{Y_{1,t}^*Y_{1,t}} &=\E{Z_t^*Z_t}+\E{Y_{1,t}^*\ee^{\ii\phi_t}}+c.c.\\
&=2c^{-2}\mathcal{R}_1\bigl(\ee^{ct}\bigr)+c^{-2}\mathcal{R}_1\bigl((ct-1)\ee^{ct}\bigr)+c.c.\\
&=c^{-2}\mathcal{R}_1\bigl((ct+1)\ee^{ct}\bigr)+c.c.\,,
\end{align*}
and hence (\ref{eq_lem_exp_var_5}).
\end{proof}

The proof of Lm.~\ref{lemma_var_Z} rests on the following two lemmas.

\begin{lemma}\label{lemma_aux_var_Z_1}
\begin{gather}
\frac{\dd}{\dd t}\left.\ES{Z_s^*Z_s}{Z_t^*Z_t}\right|_{t=s}=\ES{Z_s^*Z_s}{Z_s+Z_s^*}\,,\label{aux_Z_1_1}\\
\frac{\dd}{\dd t}\ES{Z_t^*Z_t}{Z_t^*Z_t}=2\ES{Z_t^*Z_t}{Z_t+Z_t^*}\,,\label{aux_Z_1_2}\\
\frac{\dd}{\dd t}\ES{Z_t^*Z_t}{Z_t}=\ES{\ee^{-\ii\phi_t}Z_t}{Z_t}+\ES{Z_t^*\ee^{\ii\phi_t}}{Z_t}+\ES{Z_t^*Z_t}{\ee^{\ii\phi_t}}\,,\label{aux_Z_1_3}\\
\frac{\dd}{\dd t}\ES{Z_t^*Z_t}{\ee^{\ii\phi_t-ct}}=\ES{\ee^{-\ii\phi_t}Z_t}{\ee^{\ii\phi_t-ct}}+\ES{Z_t^*\ee^{\ii\phi_t}}{\ee^{\ii\phi_t-ct}}\equiv g_{-}(t)-g_{+}(t)\,,\label{aux_Z_1_4}\\
\frac{\dd}{\dd t}\ES{Z_t^*\ee^{\ii\phi_t-ct}}{Z_t}=\ES{Z_t^*\ee^{\ii\phi_t-ct}}{\ee^{\ii\phi_t}}=- g_{+}(t)\,,\label{aux_Z_1_5}\\
\frac{\dd}{\dd t}\ES{\ee^{-\ii\phi_t-\overline{c}t}Z_t}{Z_t}=\ES{\ee^{-\ii\phi_t-\overline{c}t}Z_t}{\ee^{\ii\phi_t}}=\ee^{(c-\overline{c})t}g_{-}(t)\,,\label{aux_Z_1_6}
\end{gather}
with
\begin{equation*}
g_{\pm}(t)=\int_0^t\ee^{\pm c s}\bigl(1-\ee^{-\gamma^2 s}\bigr)\dd s\,.
\end{equation*}
\end{lemma}

The terms on the r.h.s.\ of (\ref{aux_Z_1_3}) vanish for $t=0$. As a result, the equations (\ref{aux_Z_1_4}-\ref{aux_Z_1_6}) may be rephrased as follows:
\begin{equation}\label{aux_Z_1_7}
\begin{gathered}
\ES{Z_t^*Z_t}{\ee^{\ii\phi_t}}=\ee^{ct}\int_0^tg_{-}(s)\dd s-\ee^{ct}\int_0^tg_{+}(s)\dd s\,,\\
\ES{Z_t^*\ee^{\ii\phi_t}}{Z_t}=-\ee^{ct}\int_0^t g_{+}(s)\dd s\,,\\
\ES{\ee^{-\ii\phi_t}Z_t}{Z_t}=\ee^{\overline{c}t}\int_0^t\ee^{(c-\overline{c})s}g_{-}(s)\dd s\,.
\end{gathered}
\end{equation}

The computation of the integrals yields:

\begin{lemma}\label{lemma_aux_var_Z_2}
\begin{gather}
\ee^{ct}\int_0^tg_{+}(s)\dd s=c^{-2}\ee^{2ct}-(c-\gamma^2)^{-2}\ee^{(2c-\gamma^2)t}+\frac{\gamma^2(2c-\gamma^2)}{c^2(c-\gamma^2)^2}\ee^{ct}+\frac{\gamma^2}{c(c-\gamma^2)}t\ee^{ct}\,,\label{aux_Z_2_1}\\
\ee^{ct}\int_0^tg_{-}(s)\dd s=c^{-2}-\overline{c}^{-2}\ee^{-\gamma^2t}-2\ii \frac{\gamma^2\omega}{|c|^4}\ee^{ct}-\frac{\gamma^2}{|c|^2}t\ee^{ct}\,,\label{aux_Z_2_2}\\
\ee^{\overline{c}t}\int_0^t\ee^{(c-\overline{c})s}g_{-}(s)\dd s=|c|^{-2}\bigl(1-\ee^{-\gamma^2t}\bigr)+\ii \frac{\gamma^2}{2|c|^2\omega}\bigl(\ee^{ct}-\ee^{\overline{c}t}\bigr)\,.\label{aux_Z_2_3}
\end{gather}
\end{lemma}

Let us give the proof of this computation straight away, while that of Lm.~\ref{lemma_aux_var_Z_1} will be postponed till after that of Lm.~\ref{lemma_var_Z}.
\begin{proof}[Proof of Lm. \ref{lemma_aux_var_Z_2}]
The computation rests on $\int_0^t\ee^{\alpha s}\dd s=\alpha^{-1}(\ee^{\alpha t}-1)$, which yields 
\begin{align}
\int_0^t\ee^{\alpha s}\bigl(1-\ee^{\beta s}\bigr)\dd s &=\int_0^t\bigl(\ee^{\alpha s}-\ee^{(\alpha+\beta) s}\bigr)\dd s\nonumber\\
&=\alpha^{-1}\ee^{\alpha t}-(\alpha+\beta)^{-1}\ee^{(\alpha+\beta) t}-\frac{\beta}{\alpha(\alpha+\beta)}\label{aux_var_Z_2_pf_1}\\
&=(\alpha+\beta)^{-1}\bigl(\ee^{\alpha t}-\ee^{(\alpha+\beta)t}\bigr)+\frac{\beta}{\alpha(\alpha+\beta)}\bigl(\ee^{\alpha t}-1\bigr)\label{aux_var_Z_2_pf_2}
\end{align}
by $\alpha^{-1}-(\alpha+\beta)^{-1}=\beta \alpha^{-1}(\alpha+\beta)^{-1}$; and similarly
\begin{equation}\label{aux_var_Z_2_pf_3}
\int_0^t\int_0^{s_1}\ee^{\alpha s_2}\bigl(1-\ee^{\beta s_2}\bigr)\dd s_2 \dd s_1 =\alpha^{-2}\ee^{\alpha t}-(\alpha+\beta)^{-2}\ee^{(\alpha+\beta) t}-\frac{\beta(2\alpha+\beta)}{\alpha^2(\alpha+\beta)^2}-\frac{\beta}{\alpha(\alpha+\beta)}t
\end{equation}
by $\alpha^{-2}-(\alpha+\beta)^{-2}=(2\alpha\beta+\beta^2)\alpha^{-2} (\alpha+\beta)^{-2}$. In passing we mention a slight generalization of (\ref{aux_var_Z_2_pf_2}), namely
\begin{equation}\label{aux_var_Z_2_pf_4}
\int_0^t\bigl(c_0^{-1}\ee^{\alpha s}-c_{\beta}^{-1}\ee^{(\alpha+\beta) s}\bigr)\dd s =(\alpha+\beta)^{-1}c_{\beta}^{-1}\bigl(\ee^{\alpha t}-\ee^{(\alpha+\beta)t}\bigr)+\Delta \bigl(\ee^{\alpha t}-1\bigr)
\end{equation}
with $\Delta =(\alpha c_0)^{-1}-((\alpha+\beta)c_{\beta})^{-1}$. Now (\ref{aux_Z_2_1}) follows by applying (\ref{aux_var_Z_2_pf_3}) with $\alpha=c$, $\beta=-\gamma^2$ to its integral. Likewise does (\ref{aux_Z_2_2}) with $\alpha=-c$, $\beta=-\gamma^2$ in view of 
\begin{equation*}
c+\gamma^2=-\overline{c}\,,\quad -2c-\gamma^2=-2i\omega\,.
\end{equation*}
Finally by (\ref{aux_var_Z_2_pf_1}) applied with the same $\alpha$, $\beta$ we have 
\begin{align*}
\ee^{(c-\overline{c})t}g_{-}(t) &=-c^{-1}\ee^{-\overline{c}t}+(c+\gamma^2)^{-1}\ee^{-(\overline{c}+\gamma^2)t}+\frac{\gamma^2}{c(c+\gamma^2)}\ee^{(c-\overline{c})t}\\
&=-c^{-1}\ee^{-\overline{c}t}-\overline{c}^{-1}\ee^{ct}-\frac{\gamma^2}{|c|^2}\ee^{(c-\overline{c})t}
\end{align*}
and 
\begin{equation*}
\int_0^t\ee^{(c-\overline{c})s}g_{-}(s)\dd s=|c|^{-2}\bigl(\ee^{-\overline{c}t}-\ee^{ct}\bigr)-\frac{\gamma^2}{|c|^2(c-\overline{c})}\bigl(\ee^{(c-\overline{c})t}-1\bigr)
\end{equation*}
with $c-\overline{c}=2i\omega$, proving (\ref{aux_Z_2_3}). Incidentally the last integration can be seen as an application of (\ref{aux_var_Z_2_pf_4}) with $c_0=-c$, $c_{\beta}=\overline{c}$, $\alpha=-\overline{c}$, $\alpha+\beta=c$, whence $\Delta=0$.
\end{proof}
\begin{proof}[Proof of Lm. \ref{lemma_var_Z}] 
In the expressions (\ref{aux_Z_2_1}-\ref{aux_Z_2_3}) we distinguish secular terms like $1$, $\ee^{-\gamma^2t}$ from oscillating ones such as $\ee^{\alpha t}$, $t\ee^{\alpha t}$ with $\alpha=c, 2c, \overline{c}$, including factors $\ee^{-\gamma^2 t}$. Equation (\ref{aux_Z_2_1}) only contains terms of the second type. Collecting secular terms on the r.h.s.\ of (\ref{aux_Z_1_3}) we find by (\ref{aux_Z_1_7})
\begin{align*}
S_t & \coloneqq c^{-2}-\overline{c}^{-2}\ee^{-\gamma^2t}+|c|^{-2}\bigl(1-\ee^{-\gamma^2 t}\bigr) =c^{-2}+|c|^{-2}-\bigl(\overline{c}^{-2}+|c|^{-2}\bigr)\ee^{-\gamma^2t}\\
&\,=\frac{c+\overline{c}}{|c|^4}\bigl(\overline{c}-c\ee^{-\gamma^2 t}\bigr)
\end{align*}
and $|S_t|\leq 2\gamma^2\omega^{-3}$ by $c+\overline{c}=-\gamma^2$. For later use we observe that $S_t+S_t^*$ has a better estimate:
\begin{equation}\label{lem_var_Z_pf_1}
S_t+S_t^*=\frac{(c+\overline{c})^2}{|c|^4}\bigl(1-\ee^{-\gamma^2 t}\bigr)\,,\quad |S_t+S_t^*|\leq 2\gamma^4\omega^{-4}\,.
\end{equation}
As for the oscillatory terms, we distinguish between the first two terms in (\ref{aux_Z_2_1}),
\begin{equation*}
O_t=c^{-2}\ee^{2ct}-(c-\gamma^2)^{-2}\ee^{(2c-\gamma^2)t}
\end{equation*}
and all the rest, $R_t$:
\begin{equation*}
\frac{\dd}{\dd t}\ES{Z_t^*Z_t}{Z_t}=S_t-2O_t+R_t\,.
\end{equation*}
We then integrate the equation and estimate terms as follows. 
\begin{equation*}\label{lem_var_Z_pf_2}
\left|\int_0^tR_s\dd s \right|\leq C\frac{\gamma^2}{\omega^3}\min(\omega^{-1},t)\,,
\end{equation*}
because of 
\begin{equation*}
\left|\int_0^t\ee^{w s}\dd s \right|\leq 2|w|^{-1}\min(1,|w|t)\,,\quad \left|w \int_0^ts \ee^{w s}\dd s \right|\leq C|w|^{-1}\min(1,|w|^2t^2)\,
\end{equation*}
for $\re{w}\leq 0$; the first minimum bounds the second one. By (\ref{aux_var_Z_2_pf_4}) we have 
\begin{gather*}
\int_0^tO_s\dd s=\tilde{O}_t+\Delta\bigl(\ee^{2ct}-1\bigr)\,,\quad \tilde{O}_t=(2c-\gamma^2)^{-1}(c-\gamma^2)^{-2}\bigl(\ee^{2ct}-\ee^{(2c-\gamma^2)t}\bigr)\,,
\end{gather*}
where $\Delta=(2c^3)^{-1}-(2c-\gamma^2)^{-1}(c-\gamma^2)^{-2}$ satisfies $|\Delta|\leq (3/2)|c|^{-4}\gamma^2$; moreover $|\ee^{2ct}-1|\leq 2|c|t$, $|\tilde{O}_t|\leq \omega^{-3}(1-\ee^{-\gamma^2t})\leq \omega^{-3}\gamma^2 t$. This proves (\ref{eq_lrangle_Z_t}) and 
\begin{equation*}
\ES{Z_t^*Z_t}{Z_t+Z_t^*}=-2(\tilde{O}_t+\tilde{O}_t^*)+\tilde{R}_t
\end{equation*}
with 
\begin{equation*}
|\tilde{R}_t|\leq \frac{2\gamma^4}{\omega^4}t+C\frac{\gamma^2}{\omega^4}=\frac{\gamma^2}{\omega^4}(2\gamma^2 t+C)
\end{equation*}
by (\ref{lem_var_Z_pf_1}). A further integration will give $\ES{Z_t^*Z_t}{Z_t^*Z_t}$ by (\ref{aux_Z_1_2}). That of $\tilde{O}_t$ is done by (\ref{aux_var_Z_2_pf_2}) with $\alpha=2c$, $\beta=-\gamma^2$. Both resulting terms have bounds $\omega^{-3}$ times $\omega^{-1}\gamma^2 t$. This proves (\ref{eq_llangle_Z_t}). 
For the proof of (\ref{eq_lrangle_Z_s_t}) we make use of (\ref{eq_pf_Z_t}). We so get
\begin{align*}
\langle Z_s^*Z_s; Z_t^*Z_t\rangle &=\langle Z_s^*Z_s; Z_s^*Z_s\rangle+\langle Z_s^*Z_s; Z_s^*\ee^{\ii \phi_s}\rangle\langle \hat{Z}_{t-s}\rangle+\langle Z_s^*Z_s; \ee^{-\ii \phi_s}Z_s\rangle\langle \hat{Z}_{t-s}^*\rangle\\
&=\langle Z_s^*Z_s; Z_s^*Z_s\rangle+\langle Z_sZ_s^*; Z_s\rangle\langle \hat{Z}_{t-s}\rangle+\langle Z_sZ_s^*; Z_s^*\rangle\langle \hat{Z}_{t-s}^*\rangle\,.
\end{align*}
The result now follows from (\ref{eq_llangle_Z_t}, \ref{eq_lrangle_Z_t}, \ref{eq_lem_exp_var_1}).
By definition (\ref{eq_prop_Q_t_ev_5}),
\begin{equation*}
\llangle Y_{0,t}^2\rrangle=\int_0^t\int_0^t\langle Z_{s_1}^*Z_{s_1};Z_{s_2}^*Z_{s_2}\rangle \dd s_2 \dd s_1\,,
\end{equation*}
and thus (\ref{eq_var_Y_0_t}) follows from (\ref{eq_lrangle_Z_s_t}) by integration.
\end{proof}
\begin{proof}[Proof of Lm. \ref{lemma_aux_var_Z_1}]
By straightforward differentiation and by (\ref{eq_Z_Z_star}) we get
\begin{equation*}
\frac{\dd}{\dd t}\left.\ES{Z_s^*Z_s}{Z_t^*Z_t}\right|_{t=s}=\ES{Z_s^*Z_s}{\ee^{-\ii\phi_s}Z_s+Z_s^*\ee^{\ii\phi_s}}=\ES{Z_s^*Z_s}{Z_s^*+Z_s}\,,
\end{equation*}
which is the same as (\ref{aux_Z_1_1}). Equation (\ref{aux_Z_1_2}) then follows. Equation (\ref{aux_Z_1_3}) is straightforward. Equation (\ref{aux_Z_1_4}) follows from (\ref{eq_diff_A_s_exp_t}); likewise for (\ref{aux_Z_1_5}-\ref{aux_Z_1_6}) by also noting that 
\begin{equation*}
\ES{\ee^{-\ii\phi_t}\ee^{\ii\phi_t-ct}}{\ee^{\ii\phi_t}}=\ee^{-ct}\ES{1}{\ee^{\ii\phi_t}}=0\,.
\end{equation*}
It remains to derive $g_{\pm}(t)$. Let $I=[s_1,s_2]$, ($s_1<s_2$) be an interval and set $P_I=P_{s_2}-P_{s_1}$, $\phi_I=\phi_{s_2}-\phi_{s_1}=\omega |I|+\gamma P_I$ in line with (\ref{eq_defs_U_t}). We observe that 
\begin{gather*}
\langle \ee^{\ii \phi_I};\ee^{-\ii \phi_I} \rangle = \langle \ee^{\ii \phi_I}\ee^{-\ii \phi_I} \rangle-\langle \ee^{\ii \phi_I} \rangle \langle \ee^{-\ii \phi_I} \rangle=1-\ee^{c|I|}\ee^{\overline{c}|I|}=1-\ee^{-\gamma^2|I|}\,,\\
\begin{aligned}
\langle \ee^{\ii \phi_I};\ee^{\ii \phi_I} \rangle &= \langle \ee^{2\ii \phi_I}\rangle-\langle \ee^{\ii \phi_I} \rangle^2 
=\ee^{2\ii \omega |I|}\ee^{-(2\gamma)^2|I|/2}-\ee^{2(\ii\omega-\gamma^2/2)|I|}\\
&=\ee^{2\ii \omega |I|}\bigl(\ee^{-2\gamma^2|I|}-\ee^{-\gamma^2|I|}\bigr)=\ee^{2c |I|}\bigl(\ee^{-\gamma^2|I|}-1\bigr)\,.
\end{aligned}
\end{gather*}
Now 
\begin{equation*}
\begin{aligned}
\ES{Z_t^*\ee^{\ii\phi_t}}{\ee^{\ii\phi_t}} &=\int_0^t\ES{\ee^{\ii(\phi_t-\phi_s)}}{\ee^{\ii\phi_t}}\dd s
=\int_0^t\E{\ee^{\ii\phi_s}}\ES{\ee^{\ii(\phi_t-\phi_s)}}{\ee^{\ii(\phi_t-\phi_s)}}\dd s\\
&=\int_0^t\ee^{cs}\ee^{2c(t-s)}\bigl(\ee^{-\gamma^2(t-s)}-1\bigr)\dd s
\end{aligned}
\end{equation*}
and multiplication by $\ee^{-ct}$ followed by substitution $\tau\coloneqq t-s$ gives 
\begin{equation*}
-g_{+}(t)=\int_0^t\ee^{c\tau}\bigl(\ee^{-\gamma^2\tau}-1\bigr)\dd \tau\,.
\end{equation*}
Likewise,
\begin{align*}
\ES{\ee^{-\ii\phi_t}Z_t}{\ee^{\ii\phi_t}} &=\int_0^t\ES{\ee^{-\ii(\phi_t-\phi_s)}}{\ee^{\ii\phi_t}}\dd s
=\int_0^t\E{\ee^{\ii\phi_s}}\ES{\ee^{-\ii(\phi_t-\phi_s)}}{\ee^{\ii(\phi_t-\phi_s)}}\dd s\\
&=\int_0^t\ee^{cs}\bigl(1-\ee^{-\gamma^2(t-s)}\bigr)\dd s\,,
\end{align*}
\begin{equation*}
g_{-}(t)=\int_0^t\ee^{-c\tau}\bigl(1-\ee^{-\gamma^2\tau}\bigr)\dd \tau\,.
\end{equation*}
\end{proof}
\end{section}

\begin{section}{Proof of the limiting regime}\label{pf_limiting}

The proof of Prop.~\ref{prop_ltl} calls for Brownian motion $W=(W_t)_{t\geq 0}$ in a way that is independent of the statement of the proposition itself. In fact, we shall use the Wiener-It\^o-Segal isomorphism, denoted by $\equiv$, between $\mathcal{F}=\mathcal{F}(L^2(\mathbb{R}_{+}))$ and $L^2(\Omega,\mu)$, with $\mu$ the Wiener measure on the Brownian path space $\Omega\ni W$. Any random variable $\xi(W)$ on $\Omega$ defines a multiplication operator on that $L^2$-space, an example being $\xi(W)=W_t$ for some $t\geq 0$. The above map diagonalizes the process $P=(P_t)_{t\geq 0}$, in the sense that $P_t\equiv W_t$. Moreover any such random variable, if square integrable w.r.t.\ $\mu$, also naturally defines a vector $\xi\in L^2(\Omega,\mu)$. In that sense, $\ket{\Omega}\equiv 1$, which is the constant function $1(W)=1$ on $\Omega$. In particular
\begin{equation*}
\BraKet{\Omega}{f(P_t)}{\Omega}=\int f(W_t)\dd \mu(W)\,.
\end{equation*}

Brownian motion is self-similar under diffusive scaling for any normalization of the mean square displacement. The next lemma states that the diffusive scaling limit of the process $(Z_t)_{t\geq 0}$ is Brownian motion, suitably normalized.

\begin{lemma}\label{lemma_diff_sc_Z}
We have the convergence of processes on $\Omega\ni W$:
\begin{equation}\label{eq_lem_diff_sc_Z}
 \tilde{Z}_{\varepsilon,t}(W):= \sqrt{\varepsilon}\,T_{\varepsilon}^*Z_{\varepsilon^{-1}t}(W)T_{\varepsilon}\longrightarrow -\ii\frac{\gamma}{c}B_t\,,\quad (\varepsilon\to 0)
\end{equation}
in distribution, where $B_t$ is complex Brownian motion and $T_{\varepsilon}$ is the dilation seen in (\ref{eq:tscaling}).
\end{lemma}
\begin{proof}
We have 
\begin{equation}\label{lem_Z_t_BM_2}
Z_t=\frac{1}{c}\bigl(\ee^{\ii\phi_t}-1\bigr)-\frac{\ii \gamma}{c}\int_0^t\ee^{\ii\phi_u}\dd P_u\,.
\end{equation}
Indeed, by It\^o's lemma \cite{PH84} we have 
\begin{equation}\label{lem_Z_t_BM_3}
\dd \ee^{\ii \phi_t}=\ee^{\ii\phi_t}\biggl(\ii\omega \dd t+\ii \gamma \dd P_t-\frac{1}{2}\gamma^2\dd t\biggr)=\ee^{\ii\phi_t}(c\dd t+\ii \gamma \dd P_t)\,.
\end{equation}
Integrating and solving for the first term on the r.h.s.\ yields (\ref{lem_Z_t_BM_2}). The first one on the r.h.s.\ of the latter is bounded in $t$ (by $2/\omega$) and therefore vanishes in the diffusive scaling limit. As for the integral it becomes under scaling 
\begin{equation}\label{lem_Z_t_BM_3a}
\sqrt{\varepsilon}\,T_{\varepsilon}^*\biggl(\int_0^{\varepsilon^{-1}t}\ee^{\ii \phi_u}\dd P_u \biggr)T_{\varepsilon}\overset{d}{=}\int_0^t\ee^{\ii \psi_s}\dd W_s\eqqcolon X_t
\end{equation}
with $\psi_s=\varepsilon^{-1}\omega s+\varepsilon^{-1/2}\gamma W_s$. In fact under the substitution $u\eqqcolon\varepsilon^{-1}s$, $P_u\eqqcolon \varepsilon^{-1/2}W_s$ we have $\phi_u=\psi_s$ and $W_s$ is real Brownian motion. In order to prove (\ref{eq_lem_diff_sc_Z}), it thus suffices to show 
\begin{equation*}
X_t\overset{d}{\rightarrow}B_t\,,\quad (\varepsilon \to 0)\,.
\end{equation*}
We will do so using Prohorov's theorem (e.g. \cite{S79}, Thm.\ 13.5), calling for the standard two-step procedure of proving:
\begin{enumerate}
\item The finite-dimensional distributions converge.
\item The family is tight on Wiener space.
\end{enumerate}
We begin with the first step. To this end we observe that $\dd X_t=\ee^{\ii \psi_t}\dd W_t$ satisfies
\begin{equation}\label{lem_Z_t_BM_5}
\dd X_t \dd \overline{X}_t=\dd t\,,\quad (\dd X_t)^2=\ee^{2\ii\psi_t}\dd t\,,
\end{equation}
which should be compared with $\dd B_t \dd \overline{B}_t=\dd t$, $(\dd B_t)^2=0$. The mixed product is $\dd X_t \dd W_t=\ee^{\ii \psi_t}\dd t$.

Let next $f=f(z)$ be any smooth function of $z\in \mathbb{C}$. Then $\dd f=(\partial f)\dd z+(\overline{\partial} f)\dd \overline{z}$ with $\partial=\partial/\partial z$, $\overline{\partial}=\partial/\partial \overline{z}$; thus
\begin{equation}\label{lem_Z_t_BM_6}
\dd f(X_t)=(\partial f)\dd X_t+(\overline{\partial} f)\dd \overline{X}_t+\frac{1}{2}\Bigl(2(\overline{\partial}\partial f)+\bigl(\partial^2 f\bigr)\ee^{2\ii \psi_t}+\bigl(\overline{\partial}^2 f\bigr)\ee^{-2\ii \psi_t}\Bigr)\dd t\,.
\end{equation}
We then fix times $s_i$ in $0\leq s_1 \leq \dots \leq s_n \leq t$ and consider smooth functions
\begin{equation*}
f(X_t) \equiv f(X_{s_1}, \dots, X_{s_n}, X_t)
\end{equation*}
of the values of the process $X$ at finitely many times. By induction in $n$ it suffices to show: The convergence 
\begin{equation}\label{lem_Z_t_BM_7}
\mathbb{E}[f(X_t)]\rightarrow \mathbb{E}[f(B_t)]\,,\quad (\varepsilon \to 0)\,,
\end{equation}
holds true for $t\geq s_n$ as soon as it does for $t=s_n$. Actually, by repeating the argument, it suffices to do so for $s_n\leq t\leq s_n+\Delta$ and small $\Delta >0$ (depending on $f$). Moreover and if later needed, it suffices to show (\ref{lem_Z_t_BM_7}) for 
\begin{equation}\label{lem_Z_t_BM_8}
f(X_t) =\ee^{\ii(\lambda X_t+\overline{\lambda} \overline{X}_t)}\,,\quad (\lambda\in \mathbb{C})\,, 
\end{equation}
the expectation of which is the characteristic function of $X_t$. By integrating (\ref{lem_Z_t_BM_6}) we have 
\begin{equation}\label{lem_Z_t_BM_9}
\left. \mathbb{E}[f(X_t)]\right|_{s_n}^t = \int_{s_n}^t\mathbb{E}\bigl[\bigl(\overline{\partial} \partial f\bigr)(X_s)\bigr] \dd s + \frac{1}{2}\int_{s_n}^t\mathbb{E}\Bigl[\bigl(\partial^2 f\bigr)(X_s) \ee^{2 \ii \psi_s} + \bigl(\overline{\partial}^2 f\bigr)(X_s) \ee^{- 2 \ii \psi_s}\Bigr] \dd s\,,
\end{equation}
since only the terms containing $\dd t$ contribute to the expectation. In passing we note that if $X_s$ were replaced by $B_s$ the last integral would be absent, as noted below (\ref{lem_Z_t_BM_5}). We claim that for $X_s$ it vanishes as $\varepsilon\to 0$ and more precisely that 
\begin{equation}\label{lem_Z_t_BM_10}
\left|\int_{t_0}^t \mathbb{E}\bigl[\ee^{\pm 2\ii \psi_s}g(X_s)\bigl]\dd s\right|\leq C\sqrt{\varepsilon}
\end{equation}
for $0\leq t-t_0\leq 1$ with $C$ depending on the smooth function $g=g(z)$. To see this we use the integration by parts formula 
\begin{equation}\label{lem_Z_t_BM_11}
fg|_{t_0}^t=\int_{t_0}^tf\dd g+\int_{t_0}^t g\dd f+\int_{t_0}^t (\dd f)(\dd g)
\end{equation}
with (adapted) processes $f=f_s$, $g=g_s$ and apply it to $f_s=\ee^{2\ii\psi_s}$, $g_s=g(X_s)$. Here $\dd f$ is computed from (\ref{lem_Z_t_BM_3}) by the substitution $\omega\to 2\omega \varepsilon^{-1}$, $\gamma\to 2 \gamma \varepsilon^{-1/2}$, and hence $c\to \tilde{c}\varepsilon^{-1}$ with $\tilde{c}=2(\ii \omega -\gamma^2)$, i.e.\
\begin{equation*}
\dd_t f_t=\dd \ee^{2\ii\psi_t}=\ee^{2\ii \psi_t}(\tilde{c}\varepsilon^{-1}\dd t+2\ii \gamma \varepsilon^{-1/2}\dd W_t)\,.
\end{equation*}
Moreover by (\ref{lem_Z_t_BM_6}) with $g$ in place of $f$ we have
\begin{gather*}
\dd_t g_t=(\partial g_t)\dd X_t+(\overline{\partial} g_t)\dd \overline{X}_t+\frac{1}{2}\Bigl(2(\overline{\partial}\partial g_t)+\bigl(\partial^2 g_t\bigr)\ee^{2\ii \psi_t}+\bigl(\overline{\partial}^2 g_t\bigr)\ee^{-2\ii \psi_t}\Bigr)\dd t\,,\\
(\dd_t f_t)(\dd_t g_t)=2\ii \gamma\varepsilon^{-1/2}\bigl((\partial g_t)\ee^{3\ii \psi_t}+(\overline{\partial} g_t)\ee^{\ii \psi_t}\bigr)\dd t\,.
\end{gather*}
We then take expectation values in (\ref{lem_Z_t_BM_11}). The l.h.s.\ is $O(\varepsilon^0)$ and the three integrals on the r.h.s.\ contribute at order $\varepsilon^0$, $\varepsilon^{-1}$, $\varepsilon^{-1/2}$ respectively. After multiplying by $\varepsilon\tilde{c}^{-1}$ and using $|\tilde{c}^{-1}|\leq (2\omega)^{-1}$ we so obtain
\begin{equation*}
\left|\int_{t_0}^t \mathbb{E}[\ee^{ 2\ii \psi_s}g(X_s)]\dd s\right|\leq \frac{\varepsilon}{\omega}\Vert g\Vert_{\infty}+\frac{\varepsilon}{4 \omega}\Vert D^2g\Vert_{\infty}|t-t_0|+\frac{\gamma\sqrt{\varepsilon}}{\omega}\Vert Dg\Vert_{\infty}|t-t_0|\,,
\end{equation*}
where 
\begin{equation*}
\Vert D^k g\Vert_{\infty}\coloneqq \sum_{|\alpha|=k}\bigl\Vert \overset{\text{\tiny(}-\text{\tiny )}_{\alpha}}{\partial\;} g\bigr\Vert_{\infty}\,.
\end{equation*}
Eq.\ (\ref{lem_Z_t_BM_10}) follows. Applied to (\ref{lem_Z_t_BM_9}) we get 
\begin{equation}\label{lem_Z_t_BM_12}
\left| \mathbb{E}[f(X_t)]-\mathbb{E}[f(X_{s_n})]-\int_{s_n}^t\mathbb{E}\bigl[\bigl(\overline{\partial}\partial f\bigr)(X_s)\bigr]\dd s\right|\leq C\sqrt{\varepsilon}
\end{equation}
for $\Delta\leq 1$. By the remark made there the quantity inside the modulus would vanish for $B_s$ instead of $X_s$. In order to prove (\ref{lem_Z_t_BM_7}) we may limit ourselves to (\ref{lem_Z_t_BM_8}), whence $\overline{\partial}\partial f=-|\lambda|^2f$. This prompts us to define the map $M\colon C\to C$, $h\mapsto Mh$ for $C=C([s_n,s_n+\Delta])$ given by 
\begin{equation*}
(Mh)(t)=h_0(s_n)-|\lambda|^2\int_{s_n}^th(s)\dd s\,,
\end{equation*}
where $h_0(t)=\mathbb{E}[f(B_t)]$. Clearly $M$ is a contraction
\begin{equation*}
\Vert Mh-Mh'\Vert_{\infty}\leq |\lambda|^2\Delta\Vert h-h'\Vert_{\infty}\,,
\end{equation*}
provided $|\lambda|^2\Delta<1$. Its (unique) fixed point is $h_0=Mh_0$, as remarked. Now (\ref{lem_Z_t_BM_12}) states for $h(t)\coloneqq \mathbb{E}[f(X_t)]$ that 
\begin{equation*}
h(t)=(Mh)(t)+(h(s_n)-h_0(s_n))+O(\sqrt{\varepsilon})\,,
\end{equation*}
where the middle term is $o(1)$ by the hypothesis of (\ref{lem_Z_t_BM_7}). It thus implies
\begin{equation*}
\Vert h-h_0\Vert_{\infty}\leq \Vert Mh-Mh_0\Vert_{\infty}+o(1)\leq |\lambda|^2\Delta\Vert h-h_0\Vert_{\infty}+o(1)
\end{equation*}
and so $\Vert h-h_0\Vert_{\infty}\to 0$.

We now come to the second step. By the Kolmogorov tightness criterion (see e.g.~\cite{S79}, Thm.~17.4) it suffices to show that
\begin{equation*}\label{eq_pf_kol_tight}
\mathbb{E}\bigl[|X_t - X_s|^4\bigr]\leq C |t-s|^2\,,\quad (0\leq s\leq t)\,.  
\end{equation*}
We consider $X_t-X_s$ as a process in $t$, to which we apply (\ref{lem_Z_t_BM_9}) for $f(z)=|z|^{2m}$, ($m=1,2,\dots$). Observing that the second derivatives appearing there are all bounded by a constant times $|z|^{2(m-1)}$, we obtain
\begin{equation*}
\frac{\dd}{\dd t}\mathbb{E}\bigl[|X_t - X_s|^{2m}\bigr]\leq C_m \mathbb{E}\bigl[|X_t - X_s|^{2(m-1)}\bigr]\,,
\end{equation*}
and then, recursively,
\begin{equation}\label{eq_pf_kol_tight_2}
\mathbb{E}\bigl[|X_t - X_s|^{2m}\bigr]\leq C_m' |t-s|^m\,.
\end{equation}
\end{proof}

In preparation of the proof of Prop.~\ref{prop_ltl} we rewrite (\ref{eq_prop_Q_t_ev_5}) as $X_{0,t}=Q_t+\tilde{X}_{0,t}(P)$ with $P=(P_s)_{0\leq s\leq t}$ and hence (\ref{eq_prop_Q_t_ev_2}) as 
\begin{equation}\label{def_F_t}
\begin{gathered}
U_t^*Q_tU_t=F_t(P)+Q_t\,,\\
F_t(P)=2\gamma t a^*a+X_{1,t}^*(P) a+X_{1,t}(P) a^*+\tilde{X}_{0,t}(P)=F_t(P)^*\,,
\end{gathered}
\end{equation}
where the dependence on $P$ occurs through $(Z_s)_{0\leq s\leq t}$, $Z_s=Z_s(P)$. We moreover find it convenient to restate (\ref{ltl_mcN}) for $\mathcal{M}_t\coloneqq U_t^*Q_tU_t=2\gamma t\mathcal{N}_t$ instead of $\mathcal{N}_t$. In view of $\mathcal{M}_{t\to\varepsilon^{-1}t}=2\gamma \varepsilon^{-1}t \mathcal{N}_{t\to\varepsilon^{-1}t}$ that restatement is 
\begin{equation}\label{ltl_mcN_M}
\varepsilon^2\mathcal{M}_{t\to\varepsilon^{-1}t}\to_{\ket{\Omega}} 2\gamma |\kappa|^2\int_0^{t}|B_s|^2\dd s\,.
\end{equation}
The proof itself will be carried out by means of three lemmas. The first one addresses the precise meaning of the l.h.s.\ of (\ref{eq_def_conv_ltl_expl}) in the context of (\ref{ltl_mcN_M}), yet foregoing scaling for the time being. Specifically, we are going to construct a functional calculus associated  to the family of commuting operators $(\mathcal{M}_t)_{ 0 \leq t < \infty}$.

Let $C_I=C(I,\mathbb{R})$ be the Wiener space on $I\subset \mathbb{R}_+$ and $C(C_I)$ the continuous functions on $C_I$.
\begin{lemma}\label{lemma_ltl_M_1}
For any $\ket{\psi}\in \mathcal{H}$ there is a probability measure $\mu_{\psi}$ on $C_{\mathbb{R}_+}$ and thus an operator $f(\mathcal{M})$ defined by
\begin{equation}\label{eq_lem_ltl_M_11}
\BraKet{\psi,\Omega}{f(\mathcal{M})}{\psi,\Omega}\coloneqq \int_{C_I}\dd \mu_{\psi}(\omega)f(\omega)\equiv \mathbb{E}[f]\,,\quad (f\in C(C_{\mathbb{R}_+}))\,.
\end{equation}
The map $f\mapsto f(\mathcal{M})$ extends the case where $f$ are functions of the (commuting) process at finitely many times, $f(\mathcal{M})=f(\mathcal{M}_{t_1},\dots,\mathcal{M}_{t_n})$. (The latter are defined by the functional calculus.)
\end{lemma}
\begin{proof} 
  By the Riesz-Markov theorem, the functional calculus for $f(\mathcal{M}_{t_1},\dots,\mathcal{M}_{t_n})$ defines a measure on $\mathbb{R}^n\cong \mathbb{R}^{I}$, $I=\{t_1,\dots,t_n\}$, and in fact a consistent set of such indexed by finite subsets $I\subset \mathbb{R}_+$. By the Kolmogorov extension theorem this defines a measure $\mu_{\psi}$ on some probability space that remains to be identified with $C_{\mathbb{R}_+}$. We will (a) do so for $\ket{\psi}\in \mathcal{D}$ with $\mathcal{D}\subset \mathcal{H}$ a dense subspace; then (b) (\ref{eq_lem_ltl_M_11}) defines the l.h.s.\ as a bounded quadratic form in $\ket{\psi} \in \mathcal{D}$ 
and hence $f(\mathcal{M})$ as an operator on $\mathcal{H}$. Finally (c), those operators $f(\mathcal{M})$ define a (spectral) measure $\mu_{\psi}$ on $C_{\mathbb{R}_+}$ for any $\ket{\psi}\in \mathcal{H}$, again by the Riesz-Markov theorem.

We are thus left with (a). That identification will be done for 
\begin{equation*}
\mathcal{D}=\{\ket{\psi}\in\mathcal{H}|(\psi,(a^*a)^4\psi)<\infty\}
\end{equation*}
and by means of the Kolmogorov continuity theorem (e.g. \cite{S79}, Thm.~5.1), by which it is enough to show
\begin{equation*}
\mathbb{E}\bigl[|\mathcal{M}_t-\mathcal{M}_s|^4\bigr]\leq C|t-s|^2\,,\quad (0\leq s\leq t\leq t_0)\,,
\end{equation*}
for any $t_0$, uniformly in $s,t$. The constant $C$ may depend on $t_0$. Using 
\begin{equation}\label{eq_lem_pf_ltl_M_1}
((A+B)^*(A+B))^2\leq 4 \bigl((A^*A)^2+B^*A A^* B+A^*BB^*A+(B^*B)^2\bigr)
\end{equation}
in relation with (\ref{def_F_t}) it becomes enough to establish the required bound for $\mathbb{E}[T]:=\BraKet{\psi,\Omega}{T}{\psi,\Omega}$ with $T$ any of the following operators: (i) $(\Delta Q)^4$, (ii) $(\Delta F)(\Delta Q)^2 (\Delta F)$, (iii) $(\Delta Q)(\Delta F)^2(\Delta Q)$, (iv) $(\Delta F)^4$, where 
\begin{equation*}
\Delta Q=Q_t-Q_s\,,\quad \Delta F=F_t(P)-F_s(P)\,.
\end{equation*}
Let 
\begin{equation}\label{eq_lem_pf_ltl_M_1a}
K=\ii [\Delta Q,\Delta F]=\ii [Q_t-Q_s,F_t(P)]\,.
\end{equation}
Then (ii) may be replaced by $K^*K$, besides of (iii). As for the latter, it may be replaced by $(\Delta P)(\Delta F)^2(\Delta P)$, because of $A_{[s,t]}\ket{\Omega}=0$ with $2A_{[s,t]}=\Delta Q+\ii \Delta P$. Moreover, because of commuting factors, $(\Delta P)(\Delta F)^2(\Delta P)\leq ((\Delta P)^4+(\Delta F)^4)/2$, and $(\Delta P)^4$ has the same expectation as (i). We may thus update the above list of operators to: $(\Delta Q)^4$, $(\Delta F)^4$, $K^*K$. The first one is computed easily by Wick's lemma:
\begin{equation*}
\mathbb{E}[(\Delta Q)^4]=\BraKet{\Omega}{(\Delta Q)^4}{\Omega}=3 \BraKet{\Omega}{(\Delta Q)^2}{\Omega}^2=3(\Delta t)^2\,.
\end{equation*}
Applying (\ref{eq_lem_pf_ltl_M_1}) to (\ref{def_F_t}) yields 
\begin{equation*}
(\Delta F)^4\leq C\bigl((\Delta t)^4((a^*a)^4+1)+|\Delta X_1|^4((a^*a)^2+1)+(\Delta \tilde{X}_0)^4\bigr)
\end{equation*}
with
\begin{gather*}
|\Delta X_1|^4=|C_1|^4\biggl|\int_s^tZ_{\tau}\dd \tau\biggr|^4\leq |C_1|^4\int_s^t(Z_{\tau}^*Z_{\tau})^2\dd \tau \cdot (\Delta t)^3\,,\\
(\Delta \tilde{X}_0)^4=C_2^4\biggl(\int_s^t Z_{\tau}^* Z_{\tau}\dd \tau\biggr)^4\leq C_2^4\int_s^t(Z_{\tau}^*Z_{\tau})^4\dd \tau \cdot (\Delta t)^3\,,
\end{gather*}
($C_1=2\ii \gamma \omega \alpha$, $C_2=2\gamma \omega^2 |\alpha|^2$, cf.\ (\ref{eq_prop_Q_t_ev_4}, \ref{eq_prop_Q_t_ev_5})). We then observe that
\begin{equation}\label{eq_lem_pf_ltl_M_1b}
\mathbb{E}\bigl[(Z_t^*Z_t)^2\bigr]\leq C(t^2+1)\,, \quad\mathbb{E}\bigl[(Z_t^*Z_t)^4\bigr]\leq C(t^4+1)\,,
\end{equation}
which can be seen as follows: In (\ref{lem_Z_t_BM_2}), i.e.\ $Z_t\equiv Z_t^0+X_t$, we have $|Z_t^0|\leq C$ and the moments of $X_t$ are estimated in (\ref{eq_pf_kol_tight_2}) with $s=0$. We conclude
\begin{equation*}
\mathbb{E}[(\Delta F)^4]\leq C(t_0^4+1)(\Delta t)^4\,. 
\end{equation*}
Finally, we come to $K^*K$. The drifted field
\begin{equation*}
\hat{P}_{\lambda,s}\coloneqq P_{\lambda}-2\lambda s
\end{equation*}
appears in the identity
\begin{equation}\label{eq_id_exp_Q_P}
P_s\ee^{-\ii \lambda Q_t}=\ee^{-\ii \lambda Q_t}\hat{P}_{\lambda,s}\,,\quad (0\leq s\leq t)\,,
\end{equation}
which is easily verified on the basis of (\ref{eq_comm_rel_QP}) and in turn implies
\begin{equation*}
f(P)\ee^{-\ii \lambda Q_t}=\ee^{-\ii \lambda Q_t}f(\hat{P}_{\lambda})
\end{equation*}
for functions $f$ of $(P_s)_{s\leq t}$. In particular
\begin{equation*}
\ii [Q_t, f(P)]=\frac{\dd}{\dd \lambda}\left.f(\hat{P}_{\lambda})\right|_{\lambda=0}\,.
\end{equation*}
We observe from (\ref{eq_defs_U_t}) that $\phi_t(\hat{P}_{\lambda})=(\omega-2\gamma \lambda)t+\gamma P_t$, whence 
\begin{equation*}
\ii [Q_t, Z_t(P)]=-2\ii \gamma \int_0^t s\ee^{\ii \phi_s}\dd s\,.
\end{equation*}
The commutator (\ref{eq_lem_pf_ltl_M_1a}) is computed likewise, except that the field $P_{\tau}$ drifts only for $s\leq \tau\leq t$, i.e.\
\begin{equation*}
\ii [Q_t-Q_s, f(P)]=\delta f(P)
\end{equation*}
with e.g.
\begin{equation*}
\delta Z_{\tau}=-2\ii \gamma \int_s^{\tau}(\nu-s)\ee^{\ii\phi_{\nu}}\dd \nu\,.
\end{equation*}
Clearly, $|\delta Z_{\tau}|\leq C(\Delta t)^2$. We then get
\begin{equation*}
K=\delta F=\delta X_1^*(P)a+\delta X_1(P)a^*+\delta\tilde{X}_0(P)
\end{equation*}
with 
\begin{equation*}
\delta X_1=C_1 \int_s^t\delta Z_{\tau}\dd \tau\,,\quad
\delta \tilde{X}_0=C_2\int_s^t(Z_{\tau}^*\delta Z_{\tau}+(\delta Z_{\tau})^*Z_{\tau})\dd \tau\,,
\end{equation*}
being estimated as 
\begin{gather*}
|\delta X_1|\leq C(\Delta t)^3\,,\quad |\delta \tilde{X}_0|\leq C\int_s^t |Z_{\tau}|\dd \tau \cdot (\Delta t)^2\,,\\
\bigl(\delta \tilde{X}_0\bigr)^2\leq C\int_s^t Z_{\tau}^*Z_{\tau}\dd \tau \cdot (\Delta t)^5\,.
\end{gather*}
We conclude by (\ref{eq_lem_pf_ltl_M_1b}) that 
\begin{equation*}
K^*K\leq C|\delta X_1|^2(a^*a+1)+\bigl(\delta \tilde{X}_0\bigr)^2\,,\quad
\mathbb{E}[K^*K]\leq C(t_0+1)(\Delta t)^6\,.
\end{equation*}
\end{proof}

The next lemma essentially computes functions of the observable $U_t^*Q_tU_t$ despite that its two terms on the r.h.s. of (\ref{def_F_t}) do not commute.
\begin{lemma}
Let $\hat{P}_{\lambda}$ be the drifted field $\hat{P}_{\lambda,s}=P_s-2\lambda s$. Then
\begin{equation}\label{eq_char_fct_N}
\ee^{\ii \lambda U_t^*Q_tU_t}=G_{\lambda, t}(P)\ee^{\ii \lambda F_t(P)}\ee^{\ii \lambda Q_t}
\end{equation}
where $G_{0,t}(P)=1$,
\begin{equation*}
-\ii \frac{\dd }{\dd \lambda} G_{\lambda,t}(P)=(F_t(\hat{P}_{\lambda})-F_t(P))G_{\lambda,t}(P)\,.
\end{equation*}
\end{lemma}
The lemma will be used in the scaling regime where the two terms just mentioned commute to leading order. The drift will be small, making $\hat{P}_{\lambda}$ close to $P$, and $G_{\lambda,t}(P)$ to $1$.
\begin{proof}
By the identity (\ref{eq_id_exp_Q_P}) we have
\begin{multline*}
-\ii \frac{\dd }{\dd \lambda}\bigl(\ee^{\ii \lambda U_t^*Q_tU_t}\ee^{-\ii \lambda Q_t}\ee^{-\ii \lambda F_t(P)}G_{\lambda,t}(P)^*\bigr)=\\\ee^{\ii \lambda U_t^*Q_tU_t}\bigl(\bigl(F_t(P)\ee^{-\ii \lambda Q_t}-\ee^{-\ii \lambda Q_t}F_t(P)\bigr)\ee^{-\ii \lambda F_t(P)}\\
-\ee^{-\ii \lambda Q_t}\ee^{-\ii \lambda F_t(P)}(F_t(\hat{P}_{\lambda})-F_t(P))\bigr)G_{\lambda,t}(P)^*=0\,.
\end{multline*}
\end{proof}
\begin{lemma}\label{lemma_sc_mcN}
 i) Let 
\begin{equation}\label{eq_lem_sc_mcN1}
N_{\varepsilon,t}=\varepsilon \left.T_{\varepsilon}^*N_tT_{\varepsilon}\right|_{t\to \varepsilon^{-1} t}\;,\qquad \tilde N_{\varepsilon,t}=\omega^2|\alpha|^2|\tilde{Z}_{\varepsilon,t}|^2\,,
\end{equation}
where $\tilde{Z}_{\varepsilon,t}$ is defined in (\ref{eq_lem_diff_sc_Z}). Then, for any $t\ge 0$, $f(N_{\varepsilon,t})-f(\tilde N_{\varepsilon,t})\to 0$, ($\varepsilon\to 0$) in the sense of strong convergence of operators on $\mathcal{H}\otimes \mathcal{F}$ with $\mathcal{F}\equiv L^2(\Omega,\mu)$, for any continuous bounded function $f$ on $\mathbb{R}$.

ii) Let
\begin{equation}\label{eq_lem_sc_mcN}
\mathcal{M}_{\varepsilon,t}=\varepsilon^2 \left.T_{\varepsilon}^*\mathcal{M}_tT_{\varepsilon}\right|_{t\to \varepsilon^{-1} t}\;,\qquad \mathcal{\tilde M}_{\varepsilon,t}=2\gamma\omega^2 |\alpha|^2\int_0^t|\tilde{Z}_{\varepsilon,s}|^2\dd s\,.
\end{equation}
Then $f(\mathcal{M}_{\varepsilon, t})-f(\mathcal{\tilde
  M}_{\varepsilon, t})\to 0$ in the same sense as in (i) and likewise for functions $f$ 
of the processes at finitely many times.  
\comment{ and $\mathcal{M}_{\varepsilon}=(\mathcal{M}_{\varepsilon,t})_{t\ge 0}$
and likewise for $\mathcal{\tilde M}_{\varepsilon}$. Then $f(\mathcal{M}_{\varepsilon})-$
Then $f(\mathcal{M}_{\varepsilon})$ is well-defined for $f$
  as before, but on $C(\mathbb{R}_+,\mathbb{R})$, and
  $f(\mathcal{M}_{\varepsilon})-f(\mathcal{\tilde M}_{\varepsilon})\to 0$.}
\end{lemma}
\begin{proof}
  i) In analogy with (\ref{def_F_t}) we rewrite (\ref{eq_def_A_B}, \ref{eq_prop_Q_t_ev_1}) as $N_t=(a+\ii\omega\alpha Z_t(P))^*(a+\ii\omega\alpha Z_t(P))$; we also recall (\ref{eq:tscaling}), by which we have
\begin{equation}\label{eq:dils}
T_{\varepsilon}^*Q_{\varepsilon^{-1}t}T_{\varepsilon}=\varepsilon^{-1/2}Q_t\,,\quad T_{\varepsilon}^*P_{\varepsilon^{-1}t}T_{\varepsilon}=\varepsilon^{-1/2}P_t\,,\quad
T_{\varepsilon}^*Z_{\varepsilon^{-1}t}T_{\varepsilon}=\varepsilon^{-1/2}\tilde{Z}_{\varepsilon,t}\,,
\end{equation}
and thus
\begin{equation*}  
  N_{\varepsilon,t}=\varepsilon \left.T_{\varepsilon}^*N_tT_{\varepsilon}\right|_{t\to \varepsilon^{-1} t}=\varepsilon a^*a +\ii \omega \alpha \varepsilon^{1/2}\tilde{Z}_{\varepsilon,t}a-\ii \omega\overline{\alpha}\varepsilon^{1/2}\overline{\tilde{Z}_{\varepsilon,s}}a^*+\omega^2 |\alpha|^2 |\tilde{Z}_{\varepsilon,t}|^2\,.
  \end{equation*}
  It suffices to prove the convergence for exponentials $f(x)=\ee^{\ii \lambda x}$, ($\lambda\in\mathbb{R}$). The exponential of the first three terms tends to $1$ strongly because of Lm.~\ref{lemma_diff_sc_Z}; that of the fourth term is seen in (\ref{eq_lem_sc_mcN1}).

ii) \comment{We first make sure that $f(\mathcal{M}_{\varepsilon})$ is
  well-defined. Clearly, $f(Q)$ is for $Q=(Q_t)_{t\ge 0}$, because it
  can be diagonalized over $L^2(\Omega,\mu)$, like $P$. The claim now
  follows by (\ref{eq_fldout}).}
Besides of (\ref{eq:dils}) we recall (\ref{eq_prop_Q_t_ev_4},
\ref{eq_prop_Q_t_ev_5}), by which we have
\begin{equation}
\left.T_{\varepsilon}^*Y_{1,t}T_{\varepsilon}\right|_{t\to\varepsilon^{-1}t}=\varepsilon^{-3/2}\int_0^t\tilde{Z}_{\varepsilon,s}\dd s\,,\quad
\left.T_{\varepsilon}^*Y_{0,t}T_{\varepsilon}\right|_{t\to\varepsilon^{-1}t}=\varepsilon^{-2}\int_0^t|\tilde{Z}_{\varepsilon,s}|^2\dd s\,.\label{eq_sc_Y}
\end{equation}
It suffices to prove the convergence for exponentials $f(x)=\ee^{\ii  \lambda x}$, ($\lambda\in\mathbb{R}$) with $x=\mathcal{M}_{\varepsilon,t}$ or, in case of many times, with $\sum _{i=1}^n \lambda_i x_i$ in place of $\lambda x$. But actually it suffices to do so for the process at a single time, because the exponential is multiplicative and strong convergence is inherited under multiplication. By (\ref{eq_char_fct_N}, \ref{eq_def_A_B}) with replacement $\lambda\to \tilde{\lambda}=\lambda \varepsilon^2$ we obtain
\begin{equation*}
\ee^{\ii \lambda \varepsilon^2 T_{\varepsilon}^*\mathcal{M}_tT_{\varepsilon}}=T_{\varepsilon}^*G_{\tilde{\lambda},t}(P)T_{\varepsilon} \ee^{\ii \lambda T_{\varepsilon}^*(\varepsilon^2 F_t(P))T_{\varepsilon}}\ee^{\ii \lambda T_{\varepsilon}^*(\varepsilon^2 Q_t )T_{\varepsilon}}\,.
\end{equation*}
Upon making the substitution $t\to \varepsilon^{-1} t$ dictated by (\ref{eq_lem_sc_mcN}) the third factor tends strongly to $1$, since 
\begin{equation}\label{pf_lem_sc_mcN1}
\varepsilon^2 \left.T_{\varepsilon}^*Q_tT_{\varepsilon}\right|_{t\to \varepsilon^{-1} t}=\varepsilon^{3/2}Q_t\,.
\end{equation}
As for the middle factor,
\begin{equation}\label{pf_lem_sc_mcN}
\left.\ee^{\ii \lambda T_{\varepsilon}^*(\varepsilon^2
    F_t(P))T_{\varepsilon}}\right|_{t\to \varepsilon^{-1} t}-\ee^{\ii
  \lambda \cdot 2\gamma\omega^2|\alpha|^2\int_0^t|\tilde{Z}_{\varepsilon,s}|^2\dd s}\overset{s}{\to}0\,.
\end{equation}
In fact, by (\ref{def_F_t})
\begin{equation}\label{pf_lem_sc_mcN_2}
\frac{\varepsilon^2 F_t(P)}{2\gamma}=\varepsilon^2 a^*at +\ii \omega \alpha \varepsilon^2 Y_{1,t}(P) a-\ii \omega \overline{\alpha} \varepsilon^2 \overline{Y_{1,t}(P)}a^*+\omega^2 |\alpha|^2 \varepsilon^2 Y_{0,t}(P)
\end{equation}
and hence by (\ref{eq_sc_Y})
\begin{multline*}
\left.T_{\varepsilon}^*\frac{\varepsilon^2 F_t(P)}{2\gamma}T_{\varepsilon}\right|_{t\to \varepsilon^{-1} t}=\\\varepsilon  a^*at +\ii \omega \alpha \varepsilon^{1/2} \int_0^t\tilde{Z}_{\varepsilon,s}\dd s a-\ii \omega \overline{\alpha}\varepsilon^{1/2} \int_0^t\overline{\tilde{Z}_{\varepsilon,s}}\dd s a^*+\omega^2 |\alpha|^2 \int_0^t|\tilde{Z}_{\varepsilon,s}|^2\dd s\,.
\end{multline*}
The exponential of the first three terms tends to $1$ strongly because of Lm.~\ref{lemma_diff_sc_Z}; that of the fourth term is seen in (\ref{pf_lem_sc_mcN}). It remains to show $\left.T_{\varepsilon}^*G_{\tilde{\lambda},t}(P)T_{\varepsilon}\right|_{t\to \varepsilon^{-1}t}\overset{s}{\to}1$. The generator of $G_{\tilde{\lambda},t}(P)$ w.r.t. $\lambda$ is 
\begin{equation*}
  \varepsilon^2(F_t(\hat{P}_{\tilde{\lambda}})-F_t(P))\,.
\end{equation*}
By comparison with (\ref{pf_lem_sc_mcN_2}) we are led to discuss
\begin{equation}\label{pf_lem_sc_diff}
Z_t(\hat{P}_{\tilde{\lambda}})-Z_t(P)\,.
\end{equation}
Since the drift is $-2\tilde{\lambda}s=-2\lambda\varepsilon^2s$, the first term amounts to the second up to a shift of $\gamma P_s$ by $-2\lambda\gamma \varepsilon^2 s$  in (\ref{eq_defs_U_t}), or equivalently of $\omega$ by $-\lambda\gamma\varepsilon^2$. Finally the above rescaling of $G_{\tilde{\lambda},t}(P)$ calls for the difference (\ref{pf_lem_sc_diff}) to vanish when its two terms are rescaled as seen in (\ref{eq_lem_diff_sc_Z}). It does in probability for $\varepsilon \to 0$, as can be seen by the integration by parts formula (\ref{lem_Z_t_BM_2}) and by applying stochastic dominated convergence to (\ref{lem_Z_t_BM_3a}). In fact $-\lambda\gamma\varepsilon^2\cdot \varepsilon^{-1}\to 0$.
\end{proof}
\begin{proof}[Proof of Prop. \ref{prop_ltl}]
  The convergence (\ref{eq_def_conv_ltl}) claimed in (\ref{ltl_N}) is an immediate consequence of Lm.~\ref{lemma_diff_sc_Z} and \ref{lemma_sc_mcN} together with $\BraKet{\Omega}{f(T_{\varepsilon}^*X T_{\varepsilon})}{\Omega}=\BraKet{\Omega}{f(X)}{\Omega}$ by $T_\varepsilon\ket{\Omega}=\ket{\Omega}$. The convergence (\ref{eq_def_conv_ltl_expl}) for functions $f\in C(C_I)$ claimed in (\ref{ltl_mcN}) and again in (\ref{ltl_mcN_M}) follows on the same grounds for functions $f$ depending on the processes at finitely many times. For the general case, tightness of $\mathcal{M}_{\varepsilon,t}$, cf. (\ref{eq_lem_sc_mcN}), has to be shown: We observe that $\tilde{Z}_{\varepsilon,t}\overset{d}{=}\varepsilon Z_{\varepsilon^{-1}t}$, as defined in (\ref{eq_lem_diff_sc_Z}), obeys the same moments and tightness bounds as its unscaled counterpart $Z_t$. This is so because in (\ref{lem_Z_t_BM_2}) the bound on the first term improves by a factor $\varepsilon$ and the scaling of the second term, $X_t$, was already incorporated in bounds like (\ref{eq_pf_kol_tight_2}). Moreover, all terms in (\ref{pf_lem_sc_mcN1}, \ref{pf_lem_sc_mcN_2}) are as in (\ref{def_F_t}), except for $Z_t$ replaced by $\tilde{Z}_{\varepsilon,t}$ and for additional prefactors $\varepsilon^n$ ($n=3/2, 1, 1/2,0$). Hence the same tightness bounds apply.
\end{proof}

\end{section}

\begin{appendix}
\begin{section}{Appendix}
We shall derive (\ref{eq_N_T_infty_qm}, \ref{eq_N_T_infty_n_qm}). The commutation relation $[a, a^*]=1$ implies $\ii [H_0,a]=-\ii \omega a$ and thus $\ee^{\ii H_0 t}a\ee^{-\ii H_0 t}=a\ee^{-\ii \omega t}$. Since the relation remains true upon replacing $a$ by $a-\alpha$, and $a^*$ accordingly, we also have 
\begin{equation*}
\ee^{\ii H t} a \ee^{-\ii H t}=\alpha+(a-\alpha)\ee^{-\ii\omega t}\,.
\end{equation*}
By (\ref{eq_averaged_obs}) that expression has to be multiplied from the left by its adjoint and then time averaged in order to obtain $\overline{M}_T$. As a result
\begin{equation*}
\overline{M}_{\infty}=|\alpha|^2+(a^*-\overline{\alpha})(a-\alpha)\,,
\end{equation*}
because terms $\sim \ee^{\pm \ii \omega t}$ do not contribute to the limit. This proves (\ref{eq_N_T_infty_qm}), which in turn implies 
\begin{equation*}
\BraKet{n}{\overline{M}_{\infty}^2}{n}=\BraKet{n}{(N+2|\alpha|^2)^2}{n}+|\alpha|^2\BraKet{n}{a^*a+aa^*}{n}\,,
\end{equation*}
because monomials $(a^*)^la^m$ with $l\not=m$ have vanishing expectation. The first term on the r.h.s.\ equals $\BraKet{n}{\overline{M}_{\infty}}{n}^2$ and (\ref{eq_N_T_infty_n_qm}) follows.
\end{section}
\end{appendix}
\vskip 3em
\noindent
{\bf Acknowledgment.} This research was partly supported by the NCCR SwissMAP, funded by the Swiss National Science Foundation.

\end{document}